\documentclass[10pt]{article}
\usepackage{natbib}
\usepackage{amscd}
\usepackage{amsmath}
\usepackage{amsthm}
\usepackage{amsfonts}
\usepackage{amssymb}
\usepackage{graphicx}
\usepackage{url}
\usepackage{hyperref}

\newcommand {\dsty}{\displaystyle}
\newcommand{\0}{{\bm 0}}

\newcommand{\Above}[2]{\stackrel{\scriptstyle #1}{#2}}

\newcommand{\A}{\bm{A}}

\newcommand{\B}{\bm{B}}

\newcommand{\D}{\bm{D}}

\newcommand{\Fc}{{\cal F}}

\newcommand{\Ic}{{\cal I}}

\newcommand{\I}{\bm{I}}

\newcommand{\K}{\bm{K}}
\newcommand{\Lam}{\bms{\Lambda}}

\newcommand{\Lm}{\bm{L}}	% \L is a LaTeX command

\newcommand{\Mc}{{\cal M}}

\newcommand{\M}{\bm{M}}
\newcommand{\Nc}{{\cal N}}

\newcommand{\Order}{{\cal O}}

\renewcommand{\P}{\bm{P}}
\newcommand{\Pm}{\bm{P}}

\newcommand{\R}{\bm{R}}

\newcommand{\Sc}{{\cal S}}

\newcommand{\Smh}{\hat{\bm{S}}}

\newcommand{\Sm}{\bm{S}}

\newcommand{\T}{\bm{T}}

\newcommand{\Upsi}{\bms{\Upsilon}}

\newcommand{\Vc}{{\cal V}}

\newcommand{\Wm}{\bm{W}}

\newcommand{\Zc}{{\cal Z}}
\newcommand{\Zero}{{\bf 0}}

\newcommand{\alphab}{\ov{\alpha}}
\newcommand{\alphah}{\hat {\alpha}}

\newcommand{\bdes}{\begin{description}}
\newcommand{\bd}{\begin{description}}
\newcommand{\bean}{\begin{eqnarray*}}
\newcommand{\benu}{\begin{enumerate}}
\newcommand{\ben}{\begin{enumerate}}

\newcommand{\bite}{\begin{itemize}}
\newcommand{\bi}{\begin{itemize}}
\newcommand{\bms}[1]{{\boldsymbol{#1}}}
\newcommand{\bmu}{\begin{multline}}
\newcommand{\bm}[1]{{\bf #1}}

\newcommand{\del}{\bms{\delta}}

\newcommand{\diagD}[1]{{{\bm D}_{#1} }}

\newcommand{\diagm}[1]{ { \mbox{\bf diag}\matrx{#1} } }

\newcommand{\dgam}[1]{\frac{d #1}{d \gamma}}

\newcommand{\dspfrac}[2]{\frac{\displaystyle #1}{\displaystyle #2} }

\newcommand{\dv}{\bm{d}}
\newcommand{\eban}{\begin{eqnarray*}}
\newcommand{\eba}{\begin{eqnarray}}
\newcommand{\eb}{\begin{equation}}
\newcommand{\edes}{\end{description}}

\newcommand{\ed}{\end{description}}
\newcommand{\eean}{\end{eqnarray*}}
\newcommand{\eea}{\end{eqnarray}}
\newcommand{\eenu}{\end{enumerate}}
\newcommand{\een}{\end{enumerate}}
\newcommand{\ee}{\end{equation}}
\newcommand{\eite}{\end{itemize}}
\newcommand{\ei}{\end{itemize}}
\newcommand{\emu}{\end{multline}}

\newcommand{\epsi}{\bms{\epsilon}}
\newcommand{\eqdef}{:=}

\newcommand{\etav}{{\bms{\eta}}}
\newcommand{\ev}{\bm{e}}

\newcommand{\finalversion}[1]{#1}

\newcommand{\fv}{\bm{f}}

\newcommand{\gv}{\bm{g}}

\newcommand{\hide}[1]{}	%hides text
\newcommand{\intersect}{\cap}

\newcommand{\kv}{\bm{k}}
\newcommand{\la}{\!\leftarrow\!}

\newcommand{\m}{\bm{m}}

\newcommand{\matrx}[1]{{\left[ \stackrel{}{#1}\right]}}

\newcommand{\mub}{\ov{\mu}}

\newcommand{\muv}{\bms{\mu}}
\newcommand{\muvb}{{\ov{\bms{\mu}} } }

\newcommand{\ov}{\overline}

\newcommand{\prtl}[2]{\frac{\partial #1}{\partial  #2} }

\newcommand{\piv}{\bms{\pi}}
\newcommand{\pmu}[2]{\frac{\partial #1}{\partial  \mu_{#2}} }

\newcommand{\pv}{\bm{p}}

\newcommand{\qv}{\bm{q}}

\newcommand{\suchthat}{\colon}

\newcommand{\tp}{\top}	%rm better than \sf, \tt, \rm T, or \textsf{\textsc{t}}
\newcommand{\union}{\cup}

\newcommand{\vh}{\hat{v}}
\newcommand{\vvh}{{\hat{\vv}}}
\newcommand{\vv}{\bm{v}}

\newcommand{\wbh}{\widehat{\wb} } 
  
\newcommand{\wb}{\ov{w}} 
\newcommand{\wh}{\hat{w}}

\newcommand{\xh}{\hat{x}}

\newcommand{\xvh}{\hat{\x}}
\newcommand{\xv}{\bm{x}}
\newcommand{\x}{\bm{x}}

\newcommand{\yvh}{\hat{\y}}

\newcommand{\y}{\bm{y}}
\newcommand{\zh}{\hat{z}}
\newcommand{\zvh}{\hat{\z}}

\newcommand{\z}{\bm{z}}
\newfont{\gilfont}{cmsy10 scaled\magstep0}
\newcommand{\Reals}{\mathbb{R}} % reals
\newcommand{\Eqref}[1]{Equation \eqref{#1}}

\newcommand{\Secref}[1]{Section  \ref{#1}}
\newcommand{\itemref}[1]{{\ref{#1}.\!}}
\newcommand{\Qref}[1]{{Q.\ref{#1}.\!}}

\newtheorem{Result}{Result}
\newtheorem{Corollary}{Corollary}
\newtheorem{Lemma}{Lemma}
\newtheorem{Theorem}{Theorem}
\newtheorem{Definition}{Definition}

\theoremstyle{plain}
\newtheorem{Main}{Main Results}

\title{
An Evolutionary Reduction Principle for \\Mutation Rates at Multiple Loci
}
\author{Lee Altenberg
\\ \url{altenber@hawaii.edu}
\date{\today}\footnote{To appear in {\em Bulletin of Mathematical Biology}, DOI 10.1007/s11538-010-9557-9}
}

\begin{document}
\maketitle	% Here for \documentclass[10pt]{article}
\begin{abstract}
 \finalversion{A model of mutation rate evolution for multiple loci under arbitrary selection is analyzed.   Results are obtained using techniques from \citet{Karlin:1982} that overcome the weak selection constraints needed for tractability in prior studies of multilocus event models.

A multivariate form of the reduction principle is found:  reduction results at individual loci combine topologically to produce a surface of mutation rate alterations that are neutral for a new modifier allele.  New mutation rates survive if and only if they fall below this surface --- a generalization of the hyperplane found by \citet{Zhivotovsky:Feldman:and:Christiansen:1994} for a multilocus recombination modifier.  Increases in mutation rates at some loci may evolve if compensated for by decreases at other loci.  The strength of selection on the modifier scales in proportion to the number of germline cell divisions,  and increases with the number of loci affected.   Loci that do not make a difference to marginal fitnesses at equilibrium are not subject to the reduction principle, and under fine tuning of mutation rates would be expected to have higher mutation rates than loci in mutation-selection balance.

Other results include the nonexistence of `viability analogous, Hardy-Weinberg' modifier polymorphisms under multiplicative mutation, and the sufficiency of average transmission rates to encapsulate the effect of modifier polymorphisms on the transmission of loci under selection.  A conjecture is offered regarding situations, like recombination in the presence of mutation, that exhibit departures from the reduction principle. Constraints for tractability are: tight linkage of all loci, initial fixation at the modifier locus, and mutation distributions comprising transition probabilities of reversible Markov chains.} \footnote{Dedicated to the memory of Sam Karlin, whose theorems continue to bear new fruit.}

\end{abstract}

%%%%%%%%%%%%%%%%%%%%%%%%%%%%%%%%%%%%%%%%%%%%%%%
\section{Introduction}
Genetic systems have the same material basis as developmental and physiological systems --- proteins, nucleotides, regulatory sequences, gene interaction networks, and self-organizing structures and activities in the cell and organism.  For the evolution of genetic systems, however, the Darwinian paradigm of heritable variation for fitness runs into a complication:  genetic variation for heredity can change the {\em content} of an organism's contribution to the next generation without necessarily changing the {\em quantity}, i.e. the organism's fitness.  Genetic variation for heredity can therefore be selectively neutral yet still enter into the evolutionary dynamics or the population.  Models of selectively neutral genes that modify genetic transmission --- modifier genes --- have been posed and analyzed in order to understand the evolutionary forces on the genetic system itself.

The methodology of the neutral modifier model is to find out what effects a modifier allele must have on transmission in order to survive, as a function of the conditions of the population (such as the selection regime, existing genetic system, current genes in the population, population size, etc.).  Since the modifier locus is assumed to have no intrinsic effect on fitness, its differential survival requires it become associated with alleles at loci under selection that have above average fitness.  This is called \emph{induced selection} (also referred to as `secondary selection' \citep{Karlin:and:McGregor:1972:Modifier,Kondrashov:1995}).  The task of modifier theory is to find out which effects on transmission cause a modifier allele to become associated with fitter genotypes.

Any particular system can be simulated to find out the result, and a region of systems evaluated, but one cannot be sure how such results interpolate or extrapolate without analytical results.  A `complete theory' of modifier genes would be a complete classification of population conditions and modifier effects that would produce modifier allele survival.   The current state of theory is far from this complete classification due to the limitations of mathematical techniques.  Analytical results  have been obtained  only for models that are great simplifications of reality.  The relevance of their results to real systems is justifiable only by the premise that the results extend beyond the simplified models into the space of real systems.  One may argue that this premise is implicit in the use of all theoretical results from simplified models.  

This premise is always uncertain.  To lessen the uncertainty, one would like to analyze models that are ever closer to reality.  Modifier theory has a history of being extended to ever more realistic and general models.  One result that has reappeared throughout this sequence of greater realism is the \emph{Reduction Principle}:  that population near equilibrium under a balance of selection and transformation processes will evolve in the direction of reduced rates of those transformation processes.

Modifier models exhibiting the reduction principle have mostly shared one glaring departure from reality:   that only a single transforming event during reproduction occurs for the genes under selection --- i.e. a single mutation, or single crossover.  In reality, multiple transformation events are the rule during reproduction.  This paper takes one step toward greater reality by modeling modifiers of multiple events.  

%%%%%%%%
\subsection{The Reduction Principle}

In the first analyses of genetic modifiers of mutation, recombination, and migration in the 1970s, a common result kept appearing, which was that only reduced levels of mutation, recombination, or migration could evolve when populations were near equilibrium under a balance between the forces of selection and transmission. \citet{Feldman:1972} discovered the first example of this analytical \emph{reduction result} for modifiers of recombination between two loci with two alleles under viability selection, for multiplicative and symmetric viability regimes.  Subsequent studies extended the reduction result to larger and larger spaces of models, including modifiers of mutation and migration rates, large modified rates, and arbitrary viability selection regimes \citep{Karlin:and:McGregor:1972:Modifier,Feldman:and:Balkau:1973,Balkau:and:Feldman:1973,Karlin:and:McGregor:1974,Feldman:and:Krakauer:1976,Teague:1977,Feldman:Christiansen:and:Brooks:1980}.  

It so happened that during this same time period, on a seemingly unrelated topic --- how population subdivision would affect the maintenance of genetic variation --- \citet{Karlin:1976,Karlin:1982} developed two general theorems on the spectral radius of perturbations of migration-selection systems.  These theorems show how, for two different kinds of variation in migration, a greater level of `mixing' lowers the spectral radius of the stability matrix for the system, and reduces the number of alleles that exist as protected polymorphisms.  The theorems first appear, without proof, in \citet[pp. 642--647]{Karlin:1976}, and with proof as Theorems 5.1 and 5.2 in \citet{Karlin:1982}.  

The `mixing' that occurs in migration is dynamically analogous to the `mixing' of genetic information that occurs during reproduction.  \citet{Altenberg:1984} found that Karlin's Theorem 5.2 applied to the form of variation modeled in the literature that exhibited the reduction result, through use of a general representation of genetic transmission, which hides the details of the genetic system but makes explicit the form that variation in transmission takes.  Because of the generality of Theorem 5.2, its applicability meant that the reduction result could be extended to arbitrary genetic systems and processes being modified (for which recombination, mutation, and migration are special cases), arbitrary numbers of alleles and loci, and arbitrary selection regimes --- a level of generality not often attainable in population genetics theory.  

However, the tradeoff for this generality is the very specific way that the modifier gene must vary genetic transmission in order for Theorem 5.2 to apply:  the modifier gene must scale equally all transmission probabilities between different genotypes.  This is referred to as {\em linear variation} (\citealt{Altenberg:1984,Altenberg:and:Feldman:1987}).  Linear variation has the form:
\[
T(i \la j, k) = \alpha \; P(i\la j, k) \mbox{\ for } j, k \neq i,
\] 
where $T(i\la j, k)$ is the probability that parental haplotypes $j$ and $k$ produce a gamete haplotype $i$, and $\alpha$ is the parameter controlled by the modifier gene that scales the transmission rates $P(i\la j, k)$.   

Subsequent studies of the reduction principle for linear variation include \citet{Feldman:and:Liberman:1986}, \citet{Liberman:and:Feldman:1986:MMR}, \citet{Liberman:and:Feldman:1986:GRP}, \citet{Altenberg:and:Feldman:1987}, and \citet{Altenberg:2009:Linear}.

What linear variation means, in biological terms, is that during reproduction, the genotype is `hit' only once by the transformation processes, and the probability that this hit occurs is what is controlled by the modifier gene.  Being hit only once, however, is manifestly unrealistic because reproduction almost universally exhibits multiple independent transformation events, including multiple mutations, crossovers, chromosomal reassortment, transpositions, and their combinations.  

The literature has explored the realm of multiple-hit genetic transformation models to a very limited extent, but even here, important phenomena have been discovered.  These studies can be classified into two categories:  
\benu 
\item models of two {\em mixed processes}, where the modifier gene controls one transformation process, but a second, {\em different} transformation process occurs outside of its control; and
\item models of a single process that can occur multiple times among different loci; in this case, the models are all of multi-locus recombination modification \citep{Zhivotovsky:Feldman:and:Christiansen:1994,Zhivotovsky:and:Feldman:1995}.
\eenu

\subsubsection{Mixed Processes}
Mixed processes are notable in that they are where departures from the reduction principle are found in near-equilibrium populations.  The mixed process of greatest interest has been recombination in the presence of mutation \citep{Feldman:Christiansen:and:Brooks:1980,Charlesworth:1990:MSB,Otto:and:Feldman:1997,Pylkov:Zhivotovsky:and:Feldman:1998}, and the departures from the reduction result are the basis of the `deterministic mutation hypothesis' for the evolution of sex \citep{Kondrashov:1982,Kondrashov:1984}.  Other mixed processes studied include:  the evolution of recombination in the presence of migration \citep{Charlesworth:and:Charlesworth:1979,Pylkov:Zhivotovsky:and:Feldman:1998}, or segregation and syngamy (which self-fertilization exposes in the recursion) \citep{Charlesworth:Charlesworth:and:Strobeck:1979,Holsinger:and:Feldman:1983:LM}; or models of the evolution of multiple mutation processes \citep[pp. 137--151]{Altenberg:1984}, or mutation in the presence of segregation and syngamy (also exposed in the recursion by self-fertilization \citep{Holsinger:and:Feldman:1983:MM} or fertility selection \citep{Holsinger:Feldman:and:Altenberg:1986}).  

The departures from the reduction principle caused by mixed processes are summarized by the `principle of partial control':  when the modifier gene has only partial control over the transformation occurring at loci under selection, then it may be possible for the part it controls to evolve an increase in rates \citep[pp. 149, 225--228]{Altenberg:1984}.  

\subsubsection{Multiple Hit Processes}

In the majority of these models of mixed processes, the process controlled by the modifier gene still occurs as a single event during reproduction.  Multiple events under modifier control are studied in a model of recombination between multiple loci in \citet{Zhivotovsky:Feldman:and:Christiansen:1994}, \citet{Zhivotovsky:and:Feldman:1995} and \citet{Pylkov:Zhivotovsky:and:Feldman:1998}.  By assuming fitness differences near zero (i.e. weak selection), two alleles per locus, and an unlinked modifier locus, these studies obtain analytic results for a modifier gene that has arbitrary control over recombination distributions that include multiple recombination events.

The main result in \citet{Zhivotovsky:Feldman:and:Christiansen:1994} is that they find a more sophisticated reduction principle at work: a new modifier allele will increase when rare if and only if it reduces a certain \emph{weighted sum} of recombination probabilities.  Notably,  particular recombination events may evolve an increase in rates, as long as the weighted sum is decreased.  This more complex result is distinguished by the term `generalized reduction principle'.

\subsubsection{Multiple Hit Processes Under Strong Selection}

Can the constraints of weak selection and two alleles per locus in these prior studies be dropped?  That is the aim of this paper.  Techniques from the proof of Theorem 5.1 in \citet{Karlin:1982} allow one to obtain analytic results for a more general modifier model with:
\benu
\item multiple loci under selection,
\item multiple alleles at those loci,
\item arbitrary viability selection regimes of any strength,
\item arbitrary control over the rates of the multiple events, and
\item arbitrary numbers of cell divisions from zygote to gamete.
\eenu
The latter generalization --- to multiple cell divisions in the gamete line --- is novel to this study; multiple cell divisions fundamentally rule out models with linear variation, and require multiple-hit theory.

To use these techniques, however, a different set of constraints is needed:
\benu
\item the population begins fixed at the modifier locus,
\item the only type of event is mutation,
\item mutation events occur at each locus independently,
\item the mutation rates at each locus are scaled equally by the modifier locus,
\item mutation distributions must have the form of transition probabilities of reversible Markov chains, and
\item no other genetic processes occur, including recombination.
\eenu
This last constraint --- an absence of recombination --- produces the greatest distance from realism in this model.  The distance of a constraint from physical reality, however, may not reflect its distance mathematically from techniques that make it unnecessary.  This was the case with the reduction results in \citet{Altenberg:1984} and \citet{Altenberg:and:Feldman:1987}, whose proofs depend on the assumption that no recombination occurs between the modifier locus and the loci under selection; addition of a single mathematical technique allowed this constraint to be dropped in the proof of the reduction result \citep{Altenberg:2009:Linear}.  The present results, which require an absence of recombination, are similarly presented with the hope that future developments will allow this constraint to be removed.

The modifier gene here is assumed to produce linear variation for {\em single} loci, but multiple independent events occur over multiple loci.  In other words, the modifier gene scales equally the mutation probabilities between all alleles at each single locus, but the probability of multiple mutations is the product of the probabilities of the single mutations.  Furthermore, the modifier is allowed arbitrary control over the mutation rate parameter for each locus.

Under these assumptions, one loses the use of Karlin's Theorem 5.2, since it is impossible generically for the modifier gene to produce linear variation in transmission.  However, Karlin has another theorem --- 5.1 --- which applies to a different form of variation that has many more degrees of freedom (see \eqref {eq:Theorem5.1}).  And, as it turns out, the form of variation treated in Karlin's Theorem 5.1 is perfectly suited to multiple event models. 

Karlin's Theorem 5.1 --- which does not appear to have been utilized in the literature since its original publication \citep{Karlin:1982} --- here comes into its own.  Theorem 5.1 considers stochastic matrices Karlin defines as {\em symmetrizable}, which affords use of the Rayleigh-Ritz variational characterization of  the spectral radius.  While Karlin does not mention it, symmetrizable stochastic matrices are one and the same as transition matrices for time-reversible Markov chains (see Lemma \ref {Lemma:Reversible}), which are assumed for most models of mutation in phylogenetic reconstruction \citep{Squartini:and:Arndt:2008}).  In an earlier version of this paper, I used Theorem 5.1 directly, but with further consideration it turns out that the critical tools needed are actually certain steps in Karlin's proof \citep[pp. 114--116, 197--198]{Karlin:1982}.

Application of these tools to this multiple-hit model yields ---unsurprisingly--- {\bf the reduction result}.  Moreover, the result has the form of the `generalized reduction principle' delineated by \citet{Zhivotovsky:Feldman:and:Christiansen:1994}, in that mutation rates can increase at some loci provided that mutation rates decrease sufficiently at other loci.  

The weighted average of the mutation rates found by \citet{Zhivotovsky:Feldman:and:Christiansen:1994} to be the criterion for the initial increase of the modifier is shown here to actually be the linear limit of a larger object:  namely, a smooth manifold of mutation rates that divides the space of mutation rates into those that will cause a modifier to invade, and those that will cause it to go extinct.  The existence of this manifold is found to be a topological necessity from the single-locus reduction result, shown using the Intermediate Value Theorem and Implicit Function Theorem.

For clarity, the main results of the paper from \Secref{sec:Main} are previewed here:

\begin{Main}[Multivariate Reduction Principle for Symmetrizable Mutation Rates at Multiple Loci]
Consider a genetic system in which a modifier locus controls the mutation rates of a group of loci under viability selection.  Mutations occur independently among the loci under selection.  In a population near equilibrium under a stable mutation-selection balance, fixed at the modifier locus, let a new allele of the modifier locus be introduced.  The new modifier allele can change the mutation rate parameter separately for each locus, and each parameter scales equally the probability of mutations at that locus. 

Under the following constraints:
\benu
\item mutation rates at each locus range between 0 and $1/2$,
\item no recombination or other transformation process acts on the genes,
\item the mutation matrix for each locus is irreducible, and
\item is the transition matrix for some reversible Markov chain,
\eenu
then the new modifier allele will increase (decrease) in frequency at a geometric rate if, among the loci that affect the marginal fitnesses:
\begin{enumerate}
\item  it reduces (increases) the mutation rate at any locus, and does not increase (decrease) the mutation rates at any locus;
\item it increases the mutation rates for at least one locus, and decreases the mutation rates for at least one locus, and falls below (above) the neutral manifold of mutation rates that includes the mutation rates at the equilibrium.   Should the mutation rates produced by the new modifier allele fall on this neutral manifold, then it will not change frequency at a geometric rate.
\end{enumerate}
Moreover, the further that the new set of mutation rates is from the neutral manifold, the stronger is the eventual induced selection for (against) the new modifier allele, up to a maximum fitness of $\max_i \wh_i / \wbh$ for a modifier allele that eliminates all mutation.

These results hold, in the case of multicellular organisms, for arbitrary numbers of cell divisions between gamete generations.  \finalversion{The strength of selection on the modifier locus scales in proportion to the number of cell divisions in the germline, and increases with the number of loci controlled by the modifier.}
\end{Main}

The paper proceeds with an introduction to the general modifier gene model, followed by development of mathematical tools that will be used, key theorems, and finally their application to the modifier model.  It concludes with a discussion of the particular implications of the main results, a discussion on the nature of models that depart from the reduction principle, and a conjecture about departures from the reduction principle that embodies the proposed explanation and can readily be tested.

%%%%%%%%%%%%%%%%%%%%%%%%%%%%%%%%%%%%%%%%%%%%%%%
\section{The General Evolutionary Model}

I give a condensed exposition of the general modifier model developed in \citet{Altenberg:1984} and \citet{Altenberg:and:Feldman:1987}, used in \citet{Zhivotovsky:Feldman:and:Christiansen:1994}, \citet{Zhivotovsky:and:Feldman:1995}, and \citet{Pylkov:Zhivotovsky:and:Feldman:1998}, and described recently in detail in \citet{Altenberg:2009:Linear}.  

The genome is structured in two parts:  a group of loci experiencing natural selection, and, external to the group, a neutral locus that modifies their genetic transmission probabilities.   The model assumes an infinite population, random mating, non-overlapping generations, frequency-independent viability selection, sex symmetry, and no sex-linkage.  Although selection acts on diploid genotypes, the haplotype frequencies become dynamically sufficient state variables under random mating.  Haplotypes have two indices:  one for the haplotype of the loci under selection ($i, j, k$), and one for the allele at the modifier locus ($a, b, c$).  The modifier allele is assumed to be transmitted without alteration and in Mendelian proportions (no mutation nor segregation distortion), so that the only force acting upon it is from associations it forms with the loci under selection.

%%%%%%%%%%%%%%%%%%%%%%%
The recursion on the frequency of haplotypes from one generation to the next is:
\eb
\label{eq:ModifierModel}
\wb  \;  z_{ai}' =  \sum_{bjk} T_{(r)}(ai \la aj | bk) \; w_{jk}  \; z_{aj}  \; z_{bk}
\ee
where
\bd
\item[$z_{ai}$]  is the frequency of the haplotype with allele $a$ at the modifier locus, and haplotype $i$ at the loci under selection; $z_{ai}'$ is the next generation;
\item[$w_{jk} = w_{kj}$] is the fitness of diploid genotype $jk$ for the loci under selection;
\item[$\wb$] $\dsty \eqdef \sum_{abjk}w_{jk}  \; z_{aj}  \; z_{bk}$ is the mean fitness of the population.
\item[$r_{ab}$] is the probability of recombination between the modifier locus and the nearest locus under selection,
\item[$T_{(r)}(ai \la aj | bk)$] is the probability that parental haplotypes $aj$ and $bk$ produce an offspring haplotype $ai$, conditioned on the modifier allele of the offspring being $a$:
\[
T_{(r)}(ai \la aj | bk) \eqdef (1-r_{ab}) T(ai \la aj | bk) + r_{ab} T^{\Join}(ai \la ak | bj),
\]
where the probability that parental genotype $aj, bk$ produces gamete haplotype $ai$ is:
\item[$T(ai \la aj | bk),$] when no recombination occurs between the modifier and nearest locus under selection, and
\item[$T^{\Join}(ai \la ak | bj),$] when recombination occurs between the modifier and nearest locus under selection (hence $aj | bk$ becomes $ak | bj$).  If there is no position effect from the modifier locus, then $T(ai \la aj | bk) = T^{\Join} (ai \la aj | bk).$
\ed
So, $1 = \sum_i  T(ai \la aj | bk) = \sum_i  T^{\Join}(ai \la ak | bj) , \; \forall a, b, j, k$.

At this point it becomes appropriate to point out a fundamental property of genetic transmission: 

\begin{Result}[Sufficiency of the Mean Transmission Probabilities]  The transmission probabilities enter into the dynamics of the haplotype frequencies of the loci under selection solely through their population averages, regardless of any form of underlying genetic variation for the transmission probabilities. 
\end{Result}
\begin{proof}
Let
\[
v_i \eqdef \sum_{a} z_{ai} 
\]
represent the frequency of haplotype $i$ comprising the loci under selection.  The population average of the transmission probabilities experienced by the loci under selection is:
\begin{align}
\label{eq:Tbar}
\bar{T}_{(r)}(i \la j | k) &\eqdef \frac{\sum_{ab} T_{(r)}(ai \la aj | bk) \;  z_{aj}  \; z_{bk}}{\sum_{ab}  z_{aj}  \; z_{bk}} \notag \\
&= \frac{1}{ v_{j}  \; v_{k}}\sum_{ab} T_{(r)}(ai \la aj | bk)  \;  z_{aj}  \; z_{bk}.
\end{align}
The recursion on $v_i$ is thus:
\begin{align}
\wb  \;  v_i' &= \sum_{a} z_{ai}' =  \sum_{abjk} T_{(r)}(ai \la aj | bk) \; w_{jk}  \; z_{aj}  \; z_{bk} \notag \\
&=  \sum_{jk} \bar{T}_{(r)}(i \la j | k) \; w_{jk}  \; v_{j}  \; v_{k}  \label{eq:vrecursion}.
\end{align}
Hence, any modifier polymorphism enters the dynamics of $v_i$ solely through the population mean $\bar{T}_{(r)}(i \la j | k)$.  
\end{proof}

The mean transmission probabilities $\bar{T}_{(r)}(i \la j | k)$ thus behave like a sufficient statistic, in that no additional details  about $z_{ai}$ or $T_{(r)}(ai \la aj | bk)$ matter to the value of $v_i'$.  Hence $\bar{T}_{(r)}(i \la j | k)$  screens off \citep{Salmon:1971,Salmon:1984,Brandon:1982} any details of polymorphisms of the modifier locus, such as allele frequencies or linkage disequilibrium. 

It should be noted that \eqref{eq:vrecursion} cannot be used to \emph{define} the dynamics, because $\bar{T}_{(r)}(i \la j | k)$ is itself subject to change that is not definable in terms of $\{ v_i \}$.  Hence the $\{ z_{ai} \}$ are {dynamically} sufficient state variables \citep[pp. 6--8]{Lewontin:1974}, while $\{ v_{i} \}$ are not.

\subsection{Equilibrium Relations}
A population at equilibrium under \eqref{eq:ModifierModel} must satisfy the constraint for each $b$:
\eb
\label{eq:Equilibrium}
\wbh \; \zh_{bi} =  \sum_{cjk} T_{(r)}(bi \la bj | ck) \; w_{jk}  \; \zh_{bj}  \; \zh_{ck},
\ee
where $\hat{\ }$ indicates the marked variable is at an equilibrium value.  In vector form:
\eb
\label{eq:EquilibriumFixed}
\zvh_{b} =   \M_{(b)} \; \D  \; \zvh_{b},
\ee
where
\bd
\item [$\zvh_{b}$] $\eqdef (\zh_{b1} \ \zh_{b2} \cdots \zh_{bn})^\tp$ ($\scriptstyle \tp$ is the transpose),
\item [$n$] is the number of different haplotypes for the group of loci under selection,
\item [$\D$] $\eqdef \diagm{\wh_i / \wbh }_{i,j=1}^{n}, $
\item [$\wh_j$] $\dsty \eqdef \sum_{ck} \zh_{ck} w_{jk}$, and
\item [$\M_{(b)}$] $\dsty \eqdef \matrx{\sum_{c k} T_{(r)}(bi \la bj | ck) \; \frac{w_{jk}}{\wh_j}  \; \zh_{ck} }_{i,j=1}^n$.
\ed
Note that $\D$ a non-negative diagonal matrix, and $\M$ is a (column) stochastic matrix, since $\sum_i T_{(r)}(bi \la bj | ck) = 1$ for all $b, c, j, k$, hence
$$
\sum_i [\M_{(b)}]_{ij} = \sum_i \sum_{c k} T_{(r)}(bi \la bj | ck) \; \frac{w_{jk}}{\wh_j}  \; \zh_{ck} 
= \sum_{c k}  \frac{w_{jk} \; \zh_{ck} }{\wh_j} = \frac {\wh_j}{\wh_j} = 1.
$$

A perturbation of the equilibrium to $z_{bi} = \zh_{bi} + \epsilon_{bi}$ produces:
\eba
\label{eq:Perturbation}
\lefteqn{\left(\wbh + 2 \sum_{bjck} \epsilon_{bj} w_{jk}  \zh_{ck}+ \sum_{bjck} \epsilon_{bj} \epsilon_{ck}\right)  \;   (\zh_{bi} + \epsilon_{bi}')} \\
&=&  \sum_{cjk} T_{(r)}(bi \la bj | ck) \; w_{jk}  \; (\zh_{bj} + \epsilon_{bj} ) \; (\zh_{ck} + \epsilon_{ck}). \nonumber
\eea
The system \eqref{eq:Perturbation} is assumed to be stable to {\em internal} perturbations, i.e. for perturbations where $\epsilon_{bi} \neq 0$ only for $b, i \suchthat \zh_{bi} > 0$.
  
\subsection{Initial Increase of a New Modifier Allele}
The long-term evolution of genetic transmission depends on the properties that allow a new modifier allele to invade a population and be protected from extinction.  Hence the analysis focuses on perturbations of the equilibrium by rare modifier alleles, entailing $\zh_{ai} = 0$ for all $i$ for new modifier allele $a$.  Making this substitution, and ignoring all second and higher order terms in the perturbation, the linear recursion on a new modifier allele, $a$, that perturbs \eqref{eq:Equilibrium} can be represented in vector form as:
\eb
\label{eq:ExternalStability}
\epsi_a' = \M_{(a)} \; \D \; \epsi_a \ ,
\ee
where
\begin{align}
\epsi_a &\eqdef (\epsilon_{a1} \ \epsilon_{a2} \cdots \epsilon_{an} )^\tp, \\
\intertext{and}
\label{eq:Ma}
& \M_{(a)}  \eqdef   \matrx{\displaystyle \sum_{bk} T_{(r)}(ai \la aj | bk) \dspfrac{w_{jk}}{\wh_j} \zh_{bk}}_{i, j=1}^{n}  \\
&  =  (1\!\!-\!\!r) \matrx{\displaystyle \sum_{bk} T(ai \la aj | bk) \dspfrac{w_{jk}}{\wh_j} \zh_{bk}}_{i, j=1}^{n}
\!\!\!\! +   r  \matrx{\displaystyle \sum_{bk} T^{\Join}(ai \la ak | bj) \dspfrac{w_{jk}}{\wh_j} \zh_{bk}}_{i, j=1}^{n} \nonumber.
\end{align}

%\ee
Modifier allele $a$ will increase at a geometric rate when rare if and only if the spectral radius $\rho(\M_{(a)}\D)$ exceeds 1, and will decrease at a geometric rate when rare if and only if the spectral radius $\rho(\M_{(a)}\D)$ is less than 1.  Clearly, if $\D = \I$, then $\rho(\M_{(a)}\D) = \rho(\M_{(a)}) = 1$, so geometric rates of change in modifier allele frequencies require $\D \neq \I$, a situation described by saying there is a positive {\it selection potential} (\citealt[``fitness load'' p. 63]{Altenberg:1984}; \citealt{Altenberg:and:Feldman:1987}):
\eb
\label{eq:SelectionPotential}
V =  \frac{\max_i \wh_i}{ \wbh} - 1 > 0.
\ee
We know from \eqref {eq:EquilibriumFixed} that $\rho(  \M_{(b)} \D) = 1$, since $\zvh_{b} = \M_{(b)}  \D \,  \zvh_{b}$, provided $\zvh_b \geq \neq \0$ is the only nonnegative eigenvector of $\M(b) \D$.

The analysis consists of evaluating how the relationship between $\M_{(a)}$ and the matrices $\{ \M_{(b)} \}$ maps to the relationship between $\rho(\M_{(a)}\D)$ and $\rho(\M_{(b)}\D) = 1$.

%%%%%%%%%%%%%%%%%%
%%%%%%%%%%%%%%%%%%%%%%%%%%%%%%%%%%%
\subsection{Constraints for Tractability}\label{sec:Constraints}
Evaluating how the relationship between $\M_{(a)}$ and the matrices $\{ \M_{(b)} \}$ affects $\rho(\M_{(a)} \D)$ is, in general, difficult.  The addition of three constraints makes it tractable:  
\benu
\item Mutation is the only transformation process acting on the loci under selection;
\item the modifier locus is fixed on a single allele in the initial population; and 
\item the modifier locus is tightly linked to the loci under selection.
\eenu

{\bf Mutation}.  In mutation, the products from transformation of a haplotype depend on that haplotype alone, not on the haplotype from the other parent, so $T_{(r)}(ai \la aj | bk)$ can be simplified to $T_{(r)}(ai \la aj | b)$, and \eqref{eq:Ma} becomes:
\begin{align}
& \M_{(a)} =  \matrx{\displaystyle \sum_{bk} T_{(r)}(ai \la aj | b) \dspfrac{w_{jk}}{\wh_j} \zh_{bk}}_{i, j=1}^{n}\nonumber\\
&= (1-r) \matrx{\displaystyle \sum_{bk} T(ai \la aj | b) \dspfrac{w_{jk}}{\wh_j} \zh_{bk}}_{i, j=1}^{n}
+  r \matrx{\displaystyle \sum_{bk} T^{\Join}(ai \la ak | b) \dspfrac{w_{jk}}{\wh_j} \zh_{bk}}_{i, j=1}^{n} \notag
\end{align}

{\bf Fixation of the Modifier Locus}.  The sum over $k$ involves only the terms $w_{jk} \zh_{bk}$, and can be made to cancel out the $\wh_j$ term if the initial population is fixed on a single modifier allele $b$, since in that case
$$
\sum_{k} w_{jk} \zh_{bk} = \sum_{bk} w_{jk} \zh_{bk} =  \wh_j,
$$
and therefore
\eba
\label{eq:Mmut}
\M_{(a)}&=& (1-r) \matrx{\displaystyle T(ai \la aj | b) \dspfrac {\wh_j}{\wh_j} }_{i, j=1}^{n} % only because \zh_{bk} = \vh_k since $b$ is fixed !
+  r \matrx{\displaystyle \sum_{k} T^{\Join}(ai \la ak | b) \dspfrac{w_{jk}}{\wh_j} \zh_{bk}}_{i, j=1}^{n}
\nonumber\\
&=& (1-r) \, \T_{ab} + r \, \T^{\Join}_{ab} \ \diagD{\zvh_b} \Wm \ \D^{-1} / \ \wbh,
\eea
where 
\bd
\item[$\Wm$] $ \eqdef \matrx{w_{ij}}_{i,j = 1}^L$ is the matrix of fitness coefficients; 
\item[$\diagD{\zvh_b}$] represents a diagonal matrix whose diagonal entries are the entries of the vector $\zvh_b$;
\item[$\T_{ab}$] $\eqdef  \matrx{\displaystyle  T(ai \la aj | b) }_{i, j=1}^{n}$, and
%\ \ \ 
\item[$\T^{\Join}_{ab}$] $\eqdef  \matrx{\displaystyle T^{\Join}(ai \la aj | b)}_{i, j=1}^{n}$.
%\]
\ed
The matrices $\T_{ab}$ and $\T^{\Join}_{ab}$ do not depend on either the selection coefficients or the haplotype frequencies, which is a great simplification of $\M_{(a)}$.  Additional, more compelling, reasons to fix the initial population on a single modifier allele arise from the structure of the transmission matrix, described in \Secref{subsubsec:ModifierPolymorphisms}.  Hence, fixation of the initial population on modifier allele $b$ is assumed throughout the remainder of the paper.

{\bf No Recombination with the Modifier Locus}. The analysis here follows Karlin's \citeyearpar[p. 198]{Karlin:1982} application of the Rayleigh-Ritz variational characterization of the spectral radius (\citealt[pp. 172--173]{Wilkinson:1965}, \citealt[pp. 176--180]{Horn:and:Johnson:1985}), which  requires that $\M_{(a)}$ be \emph{symmetrizable} --- i.e. of the form $\Lm \Sm \R$, where $\Lm$ and $\R$ are positive diagonal matrices, and $\Sm$ is a real symmetric matrix.  

For tight linkage of the modifier gene to the loci under selection, $r=0$, so \eqref{eq:Mmut} becomes simply $\M_{(a)} = \T_{ab}$, and symmetrizable $\T_{ab}$ are readily defined.  This is treated in \Secref{sec:MathematicalTools}.  

However, for looser linkage, $r > 0$, a key step in the analysis is blocked (see footnote \ref{foot:rec} in the proof of Theorem \ref{Theorem:MultivariateMutation}).  When $r > 0$, the term 
\[
\frac{1}{\wbh} \diagD{\zvh_b} \Wm \ \D^{-1}
\]
precludes finding families of $\M_{(a)}$ that are generically symmetrizable.     For, suppose that
\[
\M_{(a)} = (1-r) \, \T_{ab} + r \, \T^{\Join}_{ab} \ \diagD{\zvh_b} \Wm \ \D^{-1} / \ \wbh = \Lm \Sm \R.
\]
Then
\[
\Sm = (1-r) \, \Lm^{-1} \T_{ab} \R^{-1} + r \, \Lm^{-1} \T^{\Join}_{ab} \ \diagD{\zvh_b} \Wm \ \D^{-1} \R^{-1} / \ \wbh .
\]
But if this expression is to be symmetric for any $r$, then both $\Lm^{-1} \T_{ab} \R^{-1}$ and \linebreak $\Lm^{-1} \T^{\Join}_{ab} \ \diagD{\zvh_b} \Wm \ \D^{-1} \R^{-1}$ must be symmetric.  The first term requires $\T_{ab}$ be symmetrizable, as before.  But the second term, to be symmetric, forces mutation rate matrix $\T^{\Join}_{ab}$ to depend on the equilibrium haplotype frequencies $\zh_{bi}$, and the marginal fitnesses $\wh_i$, which is contrary to the biological basis of mutation rates, non-generic, and not useful for understanding the selective forces on mutation rates.  

Hence, for the remainder of the analysis, it is assumed that there is tight linkage between the modifier and the loci under selection, that mutation is the sole transformation process, and the initial population begins fixed on modifier allele $b$.  Relaxation of each of these constraints would be a goal for future analytical methods.

%%%%%%%%%%%%%%%%%%%%%%%%%%%%%%
\subsection{Multilocus Mutation Structure}

The biology of mutation provides a natural structure for multiple events.  Each nucleotide is a locus for a possible mutation event.  And in multicellular organisms, each cell division in the gamete lineage provides opportunities for mutation events (ranging from approximately 9 cell divisions in nematodes, 36 in flies, to 200 in humans \citep{Lynch:Sung:etal:2008}).

Assuming that mutations occur independently at each nucleotide, and independently from one cell division to the next, the the probability of multiple events is just the product of the probabilities of each event individually.  This is a standard assumption in many phylogenetic inference models (e.g. see \citealt{Yang:and:Nielsen:2002}, \citealt{Whelan:and:Goldman:2004}).   The modifier gene is posited to rescale equally the probabilities of all single events at each locus.  So the modifier gene could be said to produce linear variation at each {\em single locus}, but not over the entire haplotype.  Multiple cell divisions between gamete generations are represented in the dynamics simply as multiple powers of the mutation matrix.

Multiple non-independent mutation events do happen in nature, however.  A mutational event may involve multiple nucleotides.  To decompose the probabilities in this case would require nested sums and products of transition matrices (an example with dinucleotide dependencies is modeled by \citet{Squartini:and:Arndt:2008} in the context of phylogenetic processes), and will not be pursued here.  

The implementation of these assumptions is as follows.  Let:
\bd
\item[$L$] be the number of loci under selection, and $\xi, \kappa \in \{1, \ldots, L\}$ index the loci;
\item[$\muv$] be an $L$-long vector of mutation rates, one rate $\mu_ \xi $ for each locus $\xi $, whose values are controlled by the modifier gene; 
\item[$\mu_\xi P^{(\xi)}_{ij}$] be the probability of mutation from allele $j$ to allele $i$ at locus $\xi$;
\item [$\Pm^{(\xi)}$] $ \eqdef \matrx{ P^{(\xi)}_{ij}}_{i,j = 1}^{\nu_\xi}$ be the $\nu_\xi \times \nu_\xi$ transition matrix representing the mutation distribution at locus $\xi$;
\item[ $\nu_\xi$] be the number of possible alleles at locus $\xi$; and together,
\item [$\M^{(\xi)}_{\mu_\xi}$] $\eqdef (1-\mu_\xi) \I^{(\xi)} + \mu_\xi \Pm^{(\xi)}  = \I^{(\xi)} + \mu_\xi ( \Pm^{(\xi)} -  \I^{(\xi)} )$ is the $\nu_\xi \times \nu_\xi$ transmission matrix for alleles at locus $\xi$.
\ed

Under the assumption that each locus mutates independently of the other loci, the transmission matrix for the entire space of haplotypes of loci under selection is represented by the Kronecker (tensor) product ($\otimes$):
\eba
\label{eq:Mdef}
\M_\muv &=& \bigotimes_{\xi=1}^L \M^{(\xi)}_{\mu_\xi}
=  \bigotimes_{\xi=1}^L [ (1-\mu_\xi) \I^{(\xi)} + \mu_\xi \Pm^{(\xi)}  ] .
\eea
I will use the terms \emph{multivariate, multiplicative} variation to refer to the way that \eqref{eq:Mdef} maps variation in $\muv$ to variation in $\M_\muv$.

%%%%%%%%%%%%%%%%%%%%%%%%%%%%%
\subsubsection{Consequences for Modifier Polymorphisms}
\label{subsubsec:ModifierPolymorphisms}

Multivariate, multiplicative variation does not allow for the elegant ``viability analogous, Hardy-Weinberg'' (VAHW) \finalversion{equilibria, $\zvh = \yvh \otimes \vvh$,} that arise in modifier models with linear variation \finalversion{($\yvh$ is the frequency vector of the modifier alleles)(\citealt{Feldman:and:Krakauer:1976}, \citealt[pp. 130--169]{Altenberg:1984}, \citealt{Feldman:and:Liberman:1986}, \citealt{Liberman:and:Feldman:1986:MMR}, \citealt{Liberman:and:Feldman:1986:GRP})}.  

In VAHW equilibria, the parameter controlled by the modifier allele behaves as if it were a viability fitness coefficient (one minus that parameter, actually).  The transmission probabilities have the  parameterized form:
\[
T(ai \la aj | bk) = T_{ \alpha_{ab}} (i \la j|k),
\]
where the modifier locus genotype $(a,b)$ enters solely through the parameter $\alpha_{ab}$.  The VAHW structure requires that population averages of the transmission rates, as in \eqref{eq:Tbar}, be expressible as $T_{\alphab } (i \la j|k)$ for some $\alphab$.  This is possible if and only if  the space of variation in transmission is convex.  

But convexity no longer holds for multiplicative variation.  For one-locus mutation with linear variation, the convexity of the space is seen by its form:
$$
\Mc \eqdef \{ (1-\mu) \I + \mu \Pm \suchthat \mu \in (0, 1/2) \},
$$
\finalversion{where, for} a set $\{ p_i \suchthat p_i \geq 0, \sum_{i} p_i  = 1\}$, one has
$$
\sum_{i} p_i \M_{\mu_i} = \M_{\mub}, \mbox{\ where \ } \mub \eqdef  \sum_{i} p_i \mu_i.
$$
For multivariate, multiplicative variation, the space of variation is,
$$
\Mc \eqdef \{ \bigotimes_{\xi=1}^L [ (1-\mu_\xi) \I^{(\xi)} + \mu_\xi \Pm^{(\xi)}  ]  \suchthat \mu_\xi \in (0, 1/2) \}
$$
If $\mu_\xi$ varies for more than one $\xi$, $\Mc$ is no longer a convex set, so population averages of $\M_\muv$ over different $\muv$ are not of the form $\M_\etav$, for any $\etav$.  

This can be seen in the simplest case where $L = 2$.  Let there be two set of different mutation rates for each locus, and take their weighted average using $p, 1-p$, $0 < p < 1$.  Now, suppose there is $\muvb$ such that:
\eban
\lefteqn{\M_\muvb = p \M_{\muv_1} + (1-p) \M_{\muv_2}} \\
& = &p [ (1- \mu^{(1)}_1) \I^{(1)} + \mu^{(1)}_1 \Pm^{(1)}  ] \otimes [ (1-\mu^{(2)}_1) \I^{(2)} + \mu^{(2)}_1 \Pm^{(2)}  ]  \\
&+ & (1-p) [ (1- \mu^{(1)}_2) \I^{(1)} + \mu^{(1)}_2 \Pm^{(1)}  ] \otimes [ (1-\mu^{(2)}_2) \I^{(2)} +\mu^{(2)}_2 \Pm^{(2)}  ] 
 \eean
(here $\mu^{(1)}_1$ and $\mu^{(1)}_2$ refers to the two mutation rates at locus 1). Equating coefficients on each matrix term,
\begin{align*}
 (1- \mub^{(1)} ) (1- \mub ^{(2)}) 
& = p  (1- \mu^{(1)}_1) (1-\mu^{(2)}_1) +   (1- p) (1- \mu^{(1)}_2) (1-\mu^{(2)}_2), \\
\mub^{(1)}  (1- \mub ^{(2)})  &= p  \mu^{(1)}_1 (1-\mu^{(2)}_1) +   (1- p)  \mu^{(1)}_2 (1-\mu^{(2)}_2) ,\\
(1- \mub^{(1)} ) \mub ^{(2)}  &= p  (1- \mu^{(1)}_1) \mu^{(2)}_1 +   (1- p)  (1- \mu^{(1)}_2)\mu^{(2)}_2, \mbox{\ and \ }\\
\mub^{(1)} \mub ^{(2)}  &= p  \mu^{(1)}_1 \mu^{(2)}_1 +   (1- p)  \mu^{(1)}_2 \mu^{(2)}_2 .
\end{align*}
The result of adding the last two equations, and adding the second and last equations:
\begin{align*}
\mub ^{(2)}  = p  \mu^{(2)}_1 +   (1- p)  \mu^{(2)}_2, \ 
\ \mub^{(1)}   = p  \mu^{(1)}_1 +   (1- p)  \mu^{(1)}_2 
\end{align*}
gives:
\begin{align*}
\mub^{(1)} \mub ^{(2)} &= p  \mu^{(1)}_1 \mu^{(2)}_1 \  +  \   (1- p)  \mu^{(1)}_2 \mu^{(2)}_2 \\
&= [ p  \mu^{(1)}_1 +   (1- p)  \mu^{(1)}_2]  \   [ p  \mu^{(2)}_1 +   (1- p)  \mu^{(2)}_2 ] 
\end{align*}
which requires either $\mu^{(1)}_1= \mu^{(1)}_2$, or $\mu^{(2)}_1 = \mu^{(2)}_2$, which leaves  only one locus with mutation rate variation, or $p=0$, or $p=1$, which is fixation of the modifier.

 Thus, when the modifier locus is polymorphic and varies the mutation rates at more than one locus, averages over the modifier alleles can no longer be summarized by the averages of the mutation rate parameters, but instead yield mean transmission matrices \eqref{eq:Tbar} that fall outside of $\Mc$.  Hence, the relation between $\M_{(a)}$ and the matrices $\{ \M_{(b)} \}$ in \eqref{eq:EquilibriumFixed} is not simple.  
 
The analysis of modifier polymorphisms for multivariate, multiplicative variation in transmission will require techniques that can handle more general spaces of variation in transmission, a topic left for another study.

%%%%%%%%%%%%%%%%%%%%%%%%%%%%%%%%%%%%
\section{Mathematical Tools}
\label{sec:MathematicalTools}
%%%%%%%%%%%%%%%%
The analysis here is made possible with the techniques used in Theorem 5.1 of \citet[pp. 114--116, 197--198]{Karlin:1982}.  The theorem is restated as follows:

%%%%%%%%
\begin{Definition} [Symmetrizable Matrices]
A square, real matrix $\A$ is called {\em symmetrizable} to a symmetric real matrix $\Sm$  if it can be represented as a product $\A = \Lm \Sm \R$, where $\Lm$ and $\R$ are positive diagonal matrices. 
\end{Definition}
\begin{Theorem}\protect{\citep[Theorem 5.1, pp. 114--116, 197--198]{Karlin:1982}}.  
\label{Theorem5.1} 
Consider a family of stochastic matrices that commute and are symmetrizable to positive definite matrices:
\eb
\label{eq:Mi}
\Fc \eqdef \{ \M_h = \Lm \Sm_h \R \colon \M_h \M_k = \M_k \M_h \},
\ee
where $\Lm$ and $\R$ are positive diagonal matrices, and each $\Sm_h$ is a positive definite symmetric real matrix.   Let $\D$ be a positive diagonal matrix.  Then\footnote{The version in \citet[pp. 642--647]{Karlin:1976} is \eqref {eq:Theorem5.1}, but \citet[Theorem 5.1, p. 116]{Karlin:1982} states strict inequality, although the proof, pp. 197--198, does not exclude equality.  Strict inequality holds provided all $\M_h$ are irreducible.} for each $\M_h, \M_k \in \Fc$:
\eb
\label{eq:Theorem5.1}
\rho(\M_h \M_k \D) \leq \rho(\M_k \D).
\ee
\end{Theorem}
Karlin's proof uses a specially crafted inner product, but here I utilize a canonical form for symmetrizable matrices:

\begin{Lemma}[Canonical Form for Symmetrizable Matrices]
\label{Lemma:Balance} %1
A symmetrizable matrix $\A = \Lm \Sm \R$ can always be represented by a single positive diagonal matrix, $\B$, and a symmetric matrix, $\Smh$, that has the same spectrum as $\A$:
\eb
\label{eq:BalanceForm}
\A = \Lm \Sm \R = \B \Smh \B^{-1},
\ee
where 
\eb
\label{eq:B}
\B = \Lm^{1/2} \; \R^{-1/2}  c
\ee
with $c > 0$ any scalar, and 
\eb\label{eq:Smh}
\Smh =  \Lm^{1/2} \; \R^{1/2} \; \Sm \; \Lm^{1/2} \;\R^{1/2}.
\ee  

Furthermore, the Jordan canonical form,  $\Smh = \K \Lam \K^\tp$, with orthogonal matrix $\K$, and real diagonal matrix $\Lam$ of the eigenvalues of $\Smh$ and $\A$, provides a canonical form for symmetrizable $\A$:
\eb
\label{eq:CanonicalForm}
\A = \Lm \Sm \R = \B \K \Lam \K^\tp \B^{-1}. 
\ee
$\B \K$ is the matrix of right eigenvectors of $\A$ (columns), and $\K^\tp \B^{-1}$ is the matrix of left eigenvectors of $\A$ (rows).  $\B$ can be made unique by setting $c$ to a normalizer $c = \min_i[ \Lm^{-1/2} \R^{1/2}]_{ii}$ which yields $\rho(\B) = 1$.
\end{Lemma}
%%%%
\begin{proof}
Verifying by substitution:
\[
\B \Smh \B^{-1} = (\Lm^{1/2} \R^{-1/2}) \;  \Lm^{1/2} \R^{1/2} \Sm  \Lm^{1/2}  \R^{1/2} \; ( \Lm^{-1/2} \R^{1/2} )
= \Lm  \Sm \R.
\]
$\Lm^{1/2}$ and $\R^{1/2}$ exist because $\Lm$ and $\R$ are positive diagonal matrices, and $\R$ and $\Lm$  (and their powers) commute because they are diagonal matrices.  $\Smh$ is symmetric by the symmetric form of $\Lm^{1/2} \; \R^{1/2} \; \Sm \; \Lm^{1/2} \;\R^{1/2}$, so its Jordan canonical form, $\Smh = \K \Lam \K^{-1}$, has orthogonal $\K$, and real diagonal $\Lam$ \citep[Theorem 4.1.5, p. 171]{Horn:and:Johnson:1985}.  

Since $\A (\B \K) = \B \K \Lam \K^\tp \B^{-1}\B \K =(\B \K) \Lam$, it can be seen that the $j$th column of $\B\K$ is a right eigenvector associated with eigenvalue $[\Lam]_{jj}$.  Similarly, $ (\K^\tp \B^{-1}) \A  =  \K^\tp \B^{-1} \B \K \Lam \K^\tp \B^{-1} = \Lam (\K^\tp \B^{-1})$, so the $i$th row of $\K^\tp \B^{-1} $ is a left eigenvector with eigenvalue $[\Lam]_{ii}$. 

Setting
\eb
\label{eq:Bnormed}
\B  = \min_i [\Lm^{-1/2} \R^{1/2}]_{ii} \ \Lm^{1/2} \R^{-1/2} = \dspfrac{1}{\max_i [\Lm^{1/2} \R^{-1/2}]_{ii} } \ \Lm^{1/2} \R^{-1/2}
\ee
gives a unique $\B$ normalized so that $\max_i [\B]_{ii} = 1$.
\end{proof}

The symmetrizable stochastic matrices considered here have the same canonical form as the transition matrices of reversible Markov chains (\citealt[p. 33]{Keilson:1979}, \citealt[p. 296]{Ababneh:Jermiin:and:Robinson:2006}).  One may ask whether they are one and the same.  Indeed they are.  An ergodic Markov chain is {reversible} if its transition matrix $\M$ is irreducible and obeys:
\eb
\label{eq:ReversibleCriterion}
\M \ \diagD{\piv} = ( \M  \  \diagD{\piv})^\tp =   \diagD{\piv} \ \M^\tp 
\ee
where $\M \piv = \piv$, \finalversion{stationary distribution of the chain, and the Perron vector of $\M$, which refers to the eigenvector associated with the eigenvalue of largest modulus, the Perron root $1$}.  $ \diagD{\piv}$ is the diagonal matrix of the entries of $\piv$ (\citealt[pp. 414--415]{Feller:1968v1}; \citealt[pp. 143--145]{Iosifescu:1980}).
\finalversion{Hence this follows}:
%%%%%%%%
\begin{Lemma}[Reversible Markov Chains]
\label{Lemma:Reversible}
An irreducible stochastic matrix is of the form $\M = \Lm \Sm \R$, with $\Lm$ and $\R$ positive diagonal matrices and $\Sm$ a symmetric matrix, if and only if it is the transition matrix of a reversible ergodic Markov chain.
\end{Lemma}
%%%%
\begin{proof}
For the `if' part, since $\M$ is the transition matrix of an ergodic Markov chain, $\M$ must be an irreducible stochastic matrix (column stochastic by convention in this paper).   It therefore has a strictly positive Perron vector, $\piv > \0$ for Perron root $1$.  

Let $\B \eqdef \diagD{\piv}^{1/2}$, and $\Smh \eqdef \B^{-1} \M \B$.  First $\Smh$ will be shown to be symmetric. Since $\M$ satisfies \eqref{eq:ReversibleCriterion} by hypothesis:
\[
\M \ \diagD{\piv} = \M \B^2 = ( \M \B^2 )^\tp =   \B ^2 \ \M^\tp 
\]
Using $\M = \B \Smh \B^{-1}$,
\[
\begin{split}
\M \ \diagD{\piv} &= \M \B^2 =  \B \Smh \B^{-1} \B^2 =  \B \Smh \B, \\
 \B ^2 \ \M^\tp & = \B ^2 \ (\B \Smh \B^{-1})^\tp = \B ^2 \ \B^{-1} \Smh^\tp \B 
 = \B \Smh^\tp \B.
\end{split}
\]
So $ \B \Smh \B =  \B \Smh^\tp \B$, hence $\Smh = \Smh^\tp$.  

Let $\R =  \B^{-2} \Lm$ for any positive diagonal matrix $\Lm$, and let symmetric matrix
\[
\Sm = \Lm^{-1/2} \; \R^{-1/2} \; \Smh \; \Lm^{-1/2} \;\R^{-1/2}.
\]
This produces the desired $\M = \Lm \Sm \R$.

For the `only if' part, given that $\M = \Lm \Sm \R$, use $\M = \B \Smh \B^{-1}$ from Lemma \ref{Lemma:Balance}, where $\B = \Lm^{1/2} \; \R^{-1/2}$ and $\Smh =  \Lm^{1/2} \; \R^{1/2} \; \Sm \; \Lm^{1/2} \;\R^{1/2}$.
Substituting:
\[
\ev^\tp \M = \ev^\tp \B \Smh \B^{-1} = \ev^\tp \iff 
\ev^\tp \B \Smh  = \ev^\tp \B  
  \iff  \Smh^\tp  \B \ev = \Smh  \B \ev = \B \ev,
\]
hence $ \B \ev$ is a Perron vector of $\Smh$.

Let $\piv$ be the right eigenvector of each $\M$, normalized so that $\ev^\tp \piv = 1$.  Then:
\[
\M \piv = \B \Smh \B^{-1} \piv = \piv \iff 
\Smh \B^{-1} \piv  =  \B^{-1} \piv,
\]
hence $\B^{-1} \piv$ is also a Perron vector of $\Smh$.  Since $\Smh $ is irreducible, the Perron vector of $\Smh$ is unique (up to scaling, $c$), therefore
\[
c \ \B \ev =  \B^{-1} \piv  \iff  \piv = c \ \B^2 \ev =   \frac{1}{\ev^\tp \B^2 \ev} \ \B^{2} \ev =  \frac{1}{\ev^\tp \Lm \R^{-1} \ev} \  \Lm \R^{-1} \ev.
\]
Note that $\D_\piv =  \B^2 \ (\ev^\tp \Lm \R^{-1} \ev)^{-1}$.  Substituting:
\[
\M \ \diagD{\piv} = ( \B \Smh \B^{-1} ) \ \B^2 \ (\ev^\tp \Lm \R^{-1} \ev)^{-1} \\
= \B \Smh \B \ (\ev^\tp \Lm \R^{-1} \ev)^{-1},
\]
which is symmetric.  Therefore, irreducible $\M = \Lm \Sm \R$ satisfies the condition for the transition matrix of a reversible Markov chain.
\end{proof}

%%%%%%%%%%%%%%%%%%%%%%%%%%%%%%%%%%%%
%%%%%%%%%%%%%%%%%%%%%%%%%%%%%%%%%%%%
\section{Results}
With these mathematical tools in place, we are ready to analyze the modifier models.  The core result is  the following theorem that the derivative of the spectral radius of the stability matrix $\M_\muv \D$ with respect to each mutation rate parameter is negative.
%%%%%%%%%%%%%%%%%%%%%%%%%%%%%%%%%%%%

%%%%%%%%%%%%%%%%%%%%%%%%%%%%%%%%%%%%
%%%%%%%%%%%%%%%%%%%%%%%%%%%%%%%%%%%%
\begin{Theorem}[Multivariate, Multiplicative Variation]
\label{Theorem:MultivariateMutation}
Consider the stochastic matrix
\eb
\label{eq:Mmuv}
\M_\muv = \bigotimes_{\xi=1}^L [ (1-\mu_\xi) \I^{(\xi)} + \mu_\xi \Pm^{(\xi)}  ],
\ee
where each $\Pm^{(\xi)}$ is a $\nu_\xi \times \nu_\xi$ transition matrix for a reversible ergodic Markov chain.

Let $\D$ be a positive diagonal matrix.  Then for every point $\muv \in (0, 1/2)^L$, the spectral radius of 
\[
\M_\muv \D = \{ \bigotimes_{\xi=1}^L [ (1-\mu_\xi) \I^{(\xi)} + \mu_\xi \Pm^{(\xi)}  ] \} \D
\]
is non-increasing in each $\mu_\xi$.

If diagonal entries
\[
 D_{\displaystyle i_1 \cdots i_\xi \cdots i_L}  \neq D_{\displaystyle i_1 \cdots i_\xi' \cdots i_L}
\]
differ for at least one pair $i_\xi, i_\xi' \in \{1, \ldots, \nu_\xi \}$, for some $ i_1 \in \{1, \ldots, \nu_1\}$,  $\ldots$, $i_{\xi-1} \in \{1, \ldots, \nu_ {\xi-1} \}$, $i_{\xi+1} \in \{1, \ldots, \nu_ {\xi+1} \}$, $\ldots$, $i_L \in \{1, \ldots, \nu_L\}$,  
then
\[
\pmu{\rho(\M_\muv \D)}{\xi} < 0 .
\]
\end{Theorem}

%%%%
\begin{proof}  The proof is presented in three sections: applying the canonical form, evaluating the derivative, and evaluating the equality case.

%%%%%%%%%%%%%%%%%%%%%%%%%%%%%%%%
\begin{flushleft}
{\bf Applying the Canonical Form:} 
\end{flushleft}
The first step is to utilize the canonical form \eqref{eq:CanonicalForm}.  Since each $\Pm^{(\xi)}$ is the transition matrix of an ergodic Markov chain, it is irreducible and thus has Perron vector $ \piv^{(\xi)} > \Zero$, hence Lemma \ref{Lemma:Reversible} and Lemma \ref{Lemma:Balance} apply.  Therefore $\Pm^{(\xi)}$ has the canonical form
\eb
\label{eq:PBSBform}
\Pm^{(\xi)} = \B^{(\xi)}  \K^{(\xi)} \Lam^{(\xi)}{ \K^{(\xi)}}^\tp {\B^{(\xi)}}^{ -1},
\ee
where

\bd 
\item[$\B^{(\xi)}$] is a positive diagonal matrix,
\item[$\K^{(\xi)}$] is orthogonal, i.e. $\K^{(\xi)} { \K^{(\xi)}}^\tp =  { \K^{(\xi)}}^\tp  \K^{(\xi)}= \I^{(\xi)}$,  and
\item[$\Lam^{(\xi)}$] is a diagonal matrix of the eigenvalues of $ \Pm^{(\xi)}$ with largest simple eigenvalue $1$.\\
\ed
Define 
\begin{align} 
\label{eq:UpsiDef} \Upsi^{(\xi)}_{\mu_\xi} & \eqdef  (1-\mu_\xi) \I^{(\xi)} + \mu_\xi \Lam^{(\xi)}.
\end{align}

For $\mu_\xi \in (0, 1/2)$, the diagonal entries of $\Upsi^{(\xi)}_{\mu_\xi}$ are all positive.  This is seen as follows:
$\Lam^{(\xi)}$ is the diagonal matrix of the eigenvalues of $\Pm^{(\xi)}$, which are all real due to symmetrizability.  Because $\Pm^{(\xi)}$ is an irreducible stochastic matrix, by Perron-Frobenius theory it has simple largest eigenvalue 1, and all other eigenvalues of modulus at most 1.  Without loss of generality, arrange the indices so that spectral radius corresponds to index $1$.  Hence:
\eb
\label{eq:LambdaBounds}
{[\Lam^{(\xi)}]}_{11} = 1, \mbox{\ \ and \ }{ [\Lam^{(\xi)}]}_{ii} \in [-1,1) \mbox{\  for all \ }i \neq 1.
\ee   
Therefore,
\[
1-\mu_\xi + \mu_\xi {[\Lam^{(\xi)}]}_{11} = 1-\mu_\xi + \mu_\xi= 1,
\]
and for $i \neq 1$:
\eb
\label{eq:UpsiBounds}
0 < 1 - 2 \mu_\xi \leq 1-\mu_\xi + \mu_\xi {[\Lam^{(\xi)}]}_{ii} < 1.
\ee

Substituting \eqref{eq:PBSBform} and \eqref {eq:UpsiDef} into \eqref{eq:Mmuv}, one gets:
\begin{align}
\label{eq:MmuvBKetc}
\M_\muv 
&= \bigotimes_{\xi=1}^L \left( (1-\mu_\xi) \I^{(\xi)} + \mu_\xi \B^{(\xi)}  \K^{(\xi)} \Lam^{(\xi)}{ \K^{(\xi)}}^\tp {\B^{(\xi)}}^{ -1}  \right) \notag \\ 
&= \bigotimes_{\xi=1}^L \left( \B^{(\xi)}  \K^{(\xi)} [ (1-\mu_\xi) \I^{(\xi)} + \mu_\xi  \Lam^{(\xi)}]{ \K^{(\xi)}}^\tp {\B^{(\xi)}}^{ -1}  \right) \notag\\ 
&= \bigotimes_{\xi=1}^L \left( \B^{(\xi)}  \K^{(\xi)} \Upsi^{(\xi)}_{\mu_\xi} { \K^{(\xi)}}^\tp {\B^{(\xi)}}^{ -1}  \right) \notag\\ 
&= (  \bigotimes_{\xi=1}^L \B^{(\xi)} )    
(  \bigotimes_{\xi=1}^L \K^{(\xi)})  
(  \bigotimes_{\xi=1}^L \Upsi^{(\xi)}_{\mu_\xi} ) (  \bigotimes_{\xi=1}^L {\K^{(\xi)})}^\tp ( \bigotimes_{\xi=1}^L {\B^{(\xi)} }^{ -1}) \notag \notag \\
&= \B  \; \K \Upsi_\muv \K^\tp \; \B^{-1}, 
\end{align}
where
\begin{align}
\B &\eqdef  \bigotimes_{\xi=1}^L \B^{(\xi)}, \notag\\
\K &\eqdef  \bigotimes_{\xi=1}^L {\K^{(\xi)} }, \notag
\intertext{and}
\label{eq:UpsiMuv} \Upsi_\muv &\eqdef \bigotimes_{\xi=1}^L  \Upsi^{(\xi)}_{\mu_\xi}  = \bigotimes_{\xi=1}^L [ (1-\mu_\xi) \I^{(\xi)} + \mu_\xi \Lam^{(\xi)}].  
\end{align}
Since $\B$, $\K$, and $\Upsi_\muv$ are all invertible, 
they can be rotated in sequence without altering the spectrum (i.e. $\rho(\A_1 \A_2 \A_3) = \rho(\A_1^{-1} \A_1 \A_2 \A_3 \A_1) = \rho(\A_2 \A_3 \A_1 )$).  The key step from \citet[Proof of Theorem 5.1, pp. 197--198]{Karlin:1982} is to rotate the terms into a symmetric form:
\begin{align}
\rho(\M_\muv \D) &= \rho( \B  \; \K \Upsi_\muv \K^\tp \; \B^{-1} \D)
= \rho( \; \K \Upsi_\muv \K^\tp \; \B^{-1} \D  \B ) \nonumber \\
&= \rho(  \Upsi_\muv \K^\tp \; \B^{-1} \D \B \K )  
=   
\rho(  \Upsi_\muv^{1/2}  \K^\tp \D \K   \Upsi_\muv^{1/2} ), \label{eq:rhoUKDKU}
\end{align}
since\footnote{\label{foot:rec}The cancellation of the diagonal matrix $\B$ is the step that is blocked when the modifier recombines with the loci under selection, giving $r > 0$ in \eqref {eq:Mmut}, in which case, $\D$ is replaced by $ [(1-r)\D + r \, \diagD{\zvh_b} \Wm  / \wbh ]$ (assuming $\T_{ab} = \T^{\Join}_{ab}$), and so instead of the symmetric term $\B^{-1} \ \D \ \B = \D$, one has $\B^{-1} [(1-r)\D + r   / \wbh \; \diagD{\zvh_b} \Wm ] \B = (1-r) \D + r\; \B^{-1} \diagD{\zvh_b} \Wm \B \,  / \wbh$, which is generically not symmetric, thus precluding use of the Rayleigh quotient at this step.  } $ \B^{-1} \D \B = \D$.

To this symmetric form one can apply the Rayleigh-Ritz variational characterization of the spectral radius (\citealt[p. 198]{Karlin:1982}, \citealt[pp. 172--173]{Wilkinson:1965}, \citealt[pp. 176--180]{Horn:and:Johnson:1985}).   The Rayleigh-Ritz formula is, for any symmetric real matrix $\A$:
\eb
\label{eq:RayleighRitz}
\rho(\A) = \sup_{\xv \neq \Zero} \frac{\xv^\tp \A \xv}{\xv^\tp \xv}
\ee
Let $\xvh(\muv)$ be a vector, constrained to the unit sphere, $\xvh(\muv)^\tp \xvh(\muv) = 1$, that maximizes  
\[
\phi(\x) \eqdef \xv ^\tp (   \Upsi_\muv^{1/2} \K^\tp \D  \K  \Upsi_\muv^{1/2} ) \xv.
\]
Then $\xvh(\muv)$ is an eigenvector satisfying: 
\eb
\label{eq:xEigenvector}
 ( \Upsi_\muv^{1/2} \K^\tp \D  \K  \Upsi_\muv^{1/2} ) \ \xvh (\muv)= \rho(\M_\muv \D) \ \xvh (\muv).
\ee

Pre-multiplying each side of $ \eqref{eq:xEigenvector} $ by $ \B \K \Upsi_\muv^{1/2}$:
\[
\begin{split}
\rho(\M_\muv \D) \  \B \K \Upsi_\muv^{1/2} \ \xvh(\muv)  
&=  \B \K \Upsi_\muv  \K^\tp \B^{-1} \D ( \B \K  \Upsi_\muv^{1/2} \  \xvh(\muv)  ) \\
& = \M_\muv \D \ ( \B \K  \Upsi_\muv^{1/2} \  \xvh(\muv)  ),
\end{split}
\]
therefore $\B \K  \Upsi_\muv^{1/2}  \xvh(\muv)$ 
 is the eigenvector of $\M_\muv \D$ associated with the spectral radius, unique since $\M_\muv \D$ is irreducible, so call it
\eb
\label{eq:BKUvvh}
 \vvh(\muv) \eqdef \B \K  \Upsi_\muv^{1/2} \ \xvh(\muv).
\ee
Since $\B$, $\K$, and $ \Upsi_\muv^{1/2}$ are all invertible, 
\eb
\label{eq:xvh}
\xvh(\muv) =  \Upsi_\muv^{-1/2} \K^\tp \B^{-1} \ \vvh(\muv).
\ee

%\pagebreak
%%%%%%%%%%%%%%%%%%%%%%%%%%%%%%%%
\begin{flushleft}
{\bf Evaluating the Derivative}:  
\end{flushleft}
Differentiating \eqref{eq:rhoUKDKU} with respect to the mutation rate $\mu_\kappa$ at the $\kappa$th locus under selection:
\eban
\pmu{}{\kappa} \rho(\M_\muv \D) 
&=&  2 \, \xvh(\muv)^\tp ( \Upsi_\muv^{1/2} \K^\tp \D  \K  \Upsi_\muv^{1/2} ) \pmu{\xvh(\muv)}{\kappa} \\
&& +  \ 2 \,  \xvh(\muv)^\tp ({ \Upsi_\muv^{1/2}} \K^\tp \D  \K  \pmu{\Upsi_\muv^{1/2}}{\kappa} ) \xvh(\muv).
\eean
As in Karlin's proof of Theorem 5.2 \citeyearpar[p. 195]{Karlin:1982}, since $\xvh(\muv)$ maximizes the quadratic function $\phi(\xv)$,
it is a critical point of $\phi(\xv)$ \citep[p. 72]{Duistermaat:Kolk:2004}, therefore $ \left.\partial \phi(\xv) / \partial \xv \right|_{\xvh (\muv)} =0$, so 
\[
\left. \prtl{\phi(\xv)}{\xv}\right|_{\xvh (\muv)}  \left. \prtl{\xvh(\muv)}{\mu_\kappa}\right|_{\muv} 
= 2 \, \xvh(\muv)^\tp ( \Upsi_\muv^{1/2} \K^\tp \D  \K  \Upsi_\muv^{1/2} ) \pmu{\xvh(\muv)}{\kappa} = 0.
\]
Using 
$$
\pmu{ \Upsi_\muv^{1/2}}{\kappa} =  \frac{1}{2} \Upsi_\muv^{-1/2} \pmu{  \Upsi_\muv}{\kappa}
$$ 
this leaves:
\eb
\label{eq:pmukRhoMD}
\pmu{}{\kappa}  \rho(\M_\muv \D) 
= \xvh(\muv)^\tp (  \Upsi_\muv^{1/2} \K^\tp \D  \K  \Upsi_\muv^{-1/2}  \pmu{ \Upsi_\muv}{\kappa} ) \xvh(\muv) 
\ee
where $\partial  \Upsi_\muv / \partial {\mu_\kappa} $ evaluates to:
\eba
\label{eq:pmukUpsi}
\pmu{}{\kappa}  \Upsi_\muv 
&=& \pmu{}{\kappa}  \left( \bigotimes_{\xi=1}^L [ (1-\mu_\xi) \I^{(\xi)} + \mu_\xi \Lam^{(\xi)}] \right) \nonumber \\
&=& \bigotimes_{\xi=1}^{\kappa - 1} [ (1-\mu_\xi) \I^{(\xi)} + \mu_\xi \Lam^{(\xi)}] \nonumber  \\
&& \otimes \ [  \Lam^{(\kappa)} - \I^{(\kappa)}] \finalversion{\otimes} \\
&& \bigotimes_{\xi=\kappa + 1}^{L} [ (1-\mu_\xi) \I^{(\xi)} + \mu_\xi \Lam^{(\xi)}], \nonumber 
\eea
and
\begin{align}
\intertext{$\displaystyle \Upsi_\muv^{-1}  \pmu{ \Upsi_\muv}{\kappa} $} \label{eq:UpU}
&= [ \bigotimes_{\xi=1}^{\kappa-1} \I^{(\xi)} ]  \otimes [  \Lam^{(\kappa)} - \I^{(\kappa)}] 
[ (1-\mu_\kappa) \I^{(\kappa)} + \mu_\kappa \Lam^{(\kappa)}]^{-1} \finalversion{\otimes [\otimes}_{\xi=\kappa + 1}^{L} \I^{(\xi)} ] \notag \\
&= [ \bigotimes_{\xi=1}^{\kappa-1} \I^{(\xi)} ]  \otimes 
\diagm{\frac{[\Lam^{(\kappa)}]_{ii} - 1}{(1-\mu_ \kappa)  + \mu_\kappa  [\Lam^{(\kappa)}]_{ii}}}_{i=1}^{\nu_\kappa} 
\finalversion{\otimes [\otimes}_{\xi=\kappa + 1}^{L} \I^{(\xi)} ].
\end{align}
Using \eqref {eq:xEigenvector} one can substitute
$% \eban
\xvh (\muv) ^\tp (   \Upsi_\muv^{1/2} \K^\tp \D  \K   ) = \rho(\M_\muv \D) \ \xvh (\muv) ^\tp \Upsi_\muv^{-1/2} 
$ %\eean
into \eqref{eq:pmukRhoMD} and obtain:
\eba
\label{eq:pmukXU}
\pmu{}{\kappa}  \rho(\M_\muv \D) 
&=&  \rho(\M_\muv \D) \ \xvh (\muv) ^\tp \Upsi_\muv^{-1/2}   \Upsi_\muv^{-1/2}  \pmu{ \Upsi_\muv}{\kappa} \xvh(\muv) \nonumber \\
&=&  \rho(\M_\muv \D) \ \xvh (\muv) ^\tp \Upsi_\muv^{-1}  \pmu{ \Upsi_\muv}{\kappa} \xvh(\muv).
\eea

From \eqref{eq:LambdaBounds}, one sees that all the terms in \eqref{eq:pmukUpsi} are positive except for the term $\Lam^{(\kappa)} - \I^{(\kappa)}$.  The term $ \Lam^{(\kappa)} - \I^{(\kappa)} $ is a diagonal matrix with $[\Lam^{(\kappa)}]_{11} - 1 = 0$, and for $i \neq 1$, negative diagonal entries, $[\Lam^{(\kappa)}]_{ii} - 1 < 0$ .

Thus for $\muv \in (0, 1/2)^L$, $\partial \Upsi_\muv / \partial{\finalversion{\mu_}\kappa}$ and $\Upsi_\muv^{-1} \ \partial  \Upsi_\muv / \partial{\finalversion{\mu_}\kappa}$ are negative semi-definite:
\eb
\label{eq:pmukUpsiNegSemiDef}
\pmu{ \Upsi_\muv}{\kappa} \leq \neq \0 \mbox{\ \ and \ \ } \Upsi_\muv^{-1}  \pmu{ \Upsi_\muv}{\kappa}  \leq \neq \0.
\ee
Therefore \eqref{eq:pmukXU} evaluates to:
\eb
\label{eq:pmukLEQzero}
\pmu{}{\kappa}  \rho(\M_\muv \D) \leq 0.
\ee
The restriction of the mutation rates to the interval $(0, 1/2)$ is justified empirically, but their motivation here is analytic.  The open interval on the $0$ side is done to avoid the technical details of derivatives on a boundary.  The open interval on the $1/2$ side is more than technical:  if a mutation rate is allowed to be $1/2$ or greater, then the terms $(1-\mu_\xi) + \mu_\xi [\Lam^{(\xi)}]_{ii}$ may be $0$ or negative, invalidating \eqref{eq:pmukUpsiNegSemiDef}.   

%%%%%%%%%%%%%%%%%%%%%%%%%%%%%%%%
\begin{flushleft}
{\bf Evaluating the Equality Case}:
\end{flushleft}
The conditions that allow equality in \eqref{eq:pmukLEQzero} are elucidated, and the work consists mostly of tracking the zeros through the equations.  Representing \eqref {eq:pmukXU} in terms of individual entries:
\begin{align}
\label{eq:pmukSumForm}
\pmu{ \rho(\M_\muv \D) }{\kappa}
=  \rho(\M_\muv \D) 
& \sum_{i_\kappa=1}^{\nu_\kappa}   \frac{[\Lam^{(\kappa)}]_{i_\kappa i_\kappa} - 1}{(1-\mu_\kappa)  + \mu_\kappa [\Lam^{(\kappa)}]_{i_\kappa i_\kappa}} \notag \\
\times & \sum_{i_1=1}^{\nu_1} \cdots \sum_{i_{\kappa -1}=1}^{\nu_ {\kappa -1}} \  \sum_{i_{\kappa +1}=1}^{\nu_ {\kappa +1}} \cdots \sum_{i_L = 1}^{\nu_L} \xh_{i_1 i_2 \cdots i_L}^2  \leq 0.
\end{align}
In order for $\partial  \rho(\M_\muv \D) / \partial \mu_\kappa = 0$, every index $i$ where $[\Upsi_\muv^{-1} \  \partial{ \Upsi_\muv}/ \partial {\kappa}]_{ii}$ is non-zero must have $\xh_i(\muv)=0$.
So, either $[\Lam^{(\kappa)}]_{i_\kappa i_\kappa} - 1 = 0$ or $\xh_{i_1 i_2 \cdots i_\kappa \cdots i_L} = 0$.  But $[\Lam^{(\kappa)}]_{i_\kappa i_\kappa} = 1$ only for $i_\kappa = 1$, hence $\partial \rho(\M_\mu \D) / \partial \mu_\kappa = 0$ if and only if
\eb
\label{eq:xcondition}
\xh_{i_1 i_2 \cdots i_\kappa \cdots i_L} = 0 \mbox{ for all } i_\kappa \neq 1 \mbox{ and all } i_1, \ldots, i_{\kappa-1}, i_{\kappa+1}, \ldots, i_L.
\ee

This condition on $\xvh(\muv)$ can be translated into a condition on $\D$.  Using \eqref{eq:xEigenvector}:
\[
\rho(\M_\muv \D) \ \xvh(\muv)    =   \Upsi_\muv^{1/2} \K^\tp \D  \K  \Upsi_\muv^{1/2}  \xvh(\muv)  
=     \Upsi_\muv^{1/2} \K^\tp \B^{-1} \D \B \K  \Upsi_\muv^{1/2}  \xvh(\muv) .
\]
Pre-multiplying each side by $[\Upsi_\muv^{1/2} \K^\tp \B^{-1}]^{-1} = \B \K \Upsi_\muv^{-1/2} $:
\eb
\label{eq:Dform}
\rho(\M_\muv \D) \  \B \K \Upsi_\muv^{-1/2}  \xvh(\muv)  =  \D ( \B \K  \Upsi_\muv^{1/2}  \xvh(\muv)  ).
\ee

Here, it becomes notationally helpful to let $\kappa$ be either the first or last index.  Since there is no actual spatial structure or consequence to the ordering of the loci in the absence of recombination, the following derivation applies to any locus $\kappa$.  Using $\kappa=L$, the $\xvh(\muv)$ satisfying \eqref{eq:xcondition} can be written as:
\eb
\label{eq:xform}
\xvh(\muv) = \yvh \otimes \matrx{\begin{array}{c}1 \\ 0 \\ \vdots \\ 0 \end{array} }_{i=1}^{\nu_L},
\ee
for some vector $\yvh$ (were $\kappa$ in the middle, trying to write $\xvh = \yvh \otimes [1 0 0 \cdots 0 ]^\tp \otimes \yvh'$ forces a Kronecker factoring of $\xvh$ into $\yvh$ and $\yvh'$, which is not implied by \eqref {eq:xcondition}).   Substitution of \eqref {eq:xform} into \eqref{eq:BKUvvh} gives:
%\eba
\begin{align}
\label{eq:BKUx}
\vvh(\muv) &= \B \K  \Upsi_\muv^{1/2}  \xvh(\muv) \nonumber \\
& =  \left[ ( \bigotimes_{\xi=1}^{L-1} \B^{(\xi)} \K^{(\xi)}  {\Upsi^{(\xi)}_{\mu_\xi}}^{1/2} ) \ \yvh \right] \otimes  \B^{(L)} \K^{(L)}  {\Upsi^{(L)}_{\mu_L}}^{1/2} \ [1 0 \cdots 0]^\tp \nonumber \\
&= \left[ ( \bigotimes_{\xi=1}^{L-1} \B^{(\xi)} \K^{(\xi)}  {\Upsi^{(\xi)}_{\mu_\xi}}^{1/2} ) \ \yvh \right] \otimes  \B^{(L)} [\K^{(L)}]_1 ,
\end{align}
%\eea
where $[\K^{(L)}]_1$ is the first column of $\K^{(L)}$, since
\begin{align}
\label{eq:BKUform}
\B^{(L)} \K^{(L)}  {\Upsi^{(L)}_{\mu_L}}^{1/2} \ [1 0 \cdots 0]^\tp
&= \B^{(L)} \K^{(L)} \ \left[ [{\Upsi^{(L)}_{\mu_L}}]_{11}^{1/2} \  0 \cdots 0 \right]^\tp \notag \\
&= \B^{(L)} \K^{(L)} \ [1 0 \cdots 0]^\tp 
= \B^{(L)} [ \K^{(L)}]_1.
\end{align}
By construction, $\B^{(L)}[\K^{(L)}]_1$ is the Perron vector of irreducible $\P^{(L)}$, hence $\B^{(L)}[\K^{(L)}]_1 > \0$.

Substituting \eqref {eq:BKUx} and \eqref {eq:BKUform} into \eqref{eq:Dform} gives:
\begin{multline}
\label{eq:Dform2}
%\begin{split}
\D  \left(  \left[ \{ \bigotimes_{\xi=1}^{L-1} \B^{(\xi)} \K^{(\xi)}  {\Upsi^{(\xi)}_{\mu_\xi}}^{1/2} \} \ \yvh \right] 
\otimes \B^{(L)} [\K^{(L)}]_1  \right) \\
=   \rho(\M_\muv \D) \  \left[ ( \bigotimes_{\xi=1}^{L-1} \B^{(\xi)} \K^{(\xi)}  {\Upsi^{(\xi)}_{\mu_\xi}}^{-1/2} ) \ \yvh \right] 
\otimes   \B^{(L)} [\K^{(L)}]_1 .
%\end{split}
\end{multline}
\Eqref{eq:Dform2} becomes clearer if the terms are represented by single symbols.  Let
\bd
\item[$\fv$] $\displaystyle \eqdef  ( \bigotimes_{\xi=1}^{L-1}\B^{(\xi)}  \K^{(\xi)}  {\Upsi^{(\xi)}_{\mu_\xi}}^{1/2} ) \ \yvh,$  
\item[$\gv$] $\displaystyle\eqdef  ( \bigotimes_{\xi=1}^{L-1}\B^{(\xi)}  \K^{(\xi)}  {\Upsi^{(\xi)}_{\mu_\xi}}^{-1/2} ) \ \yvh,$  
\item[$\kv$] $\displaystyle\eqdef  \B^{(L)} {[\K^{(L)}]}_1,$ and 
\item[$\rho$] $\displaystyle \eqdef \rho(\M_\muv \D).$
\ed
Then 
\[
\vvh = \fv \otimes \kv,
\]
 and  \eqref {eq:Dform2} becomes:
\eb
\label{eq:Dform3}
\D (\fv \otimes \kv) = \rho \ \gv \otimes \kv.
\ee
Now \eqref {eq:Dform2} 
may be expressed in terms of the entries:  Let $i_1$ index the haplotypes of all loci except $L$, and $i_2$ index the alleles of locus $L$.  Then \eqref {eq:Dform3} is represented as:
\eb
\label{eq:Dform4}
D_{i_1 i_2} \ f_{i_1} k_{i_2} = \rho \ g_{i_1} \ k_{i_2}
\ee
for each $i_1, i_2$.   Since $k_{i_2} > 0$, this implies $D_{i_1 i_2} \ f_{i_1} = \rho \ g_{i_1}$ for all $i_2$.
Because $\M_\muv \D$ is irreducible, 
\eb
\label{vvhGTzero}
\vvh > \0,
\ee
therefore $f_{i_1} > 0$.  Then $ D_{i_1 i_2} = \rho \ g_{i_1}/ \ f_{i_1}$ for all $i_2$.   Thus the chain of implications that starts with $\partial \rho(\M_\muv \D) / \partial \mu_L = 0 $ concludes with the finding that $\D$ must be of the form
\[
\D = \ov{\D} \otimes \I^{(L)}
\] 
where $\ov{\D} = \diagm{D_{i_1}}$.

To summarize the equality case, recall that the above derivation applies with respect to any locus $\kappa$.  Thus, 
if and only if 
\eb
\label{eq:Differ}
D_{\displaystyle i_1 \cdots i_\kappa \cdots i_L}  \neq D_{\displaystyle i_1 \cdots i_\kappa' \cdots i_L}
\ee
for at least one pair $i_\kappa, i_\kappa' \in \{1, \ldots, \nu_\kappa \}$, for some $ i_1 \in \{1, \ldots, \nu_1\}$,  $\ldots$, $i_{\xi-1} \in \{1, \ldots, \nu_ {\xi-1} \}$, $i_{\xi+1} \in \{1, \ldots, \nu_ {\xi+1} \}$, $\ldots$, $i_L \in \{1, \ldots, \nu_L\}$,  
 then
\[
\pmu{\rho(\M_\muv \D)}{\kappa} < 0 . \qedhere
\]
\end{proof}

{\bf Remarks}.  In the case where $[\D]_{ii} = 0$ for some $i$, $\M_\muv \D$ is no longer irreducible, and so a unique positive Perron vector $\vvh$ for $\M_\muv \D$ is no longer guaranteed.  If the set of $[\D]_{ii}=0$ entries dissects the haplotype space into multiple non-communicating sub-spaces, each of these is represented by an isolated block in the Frobenius normal form of $\M_\muv \D$, thus $\M_\muv \D$ will have multiple non-negative eigenvectors.  This situation is more complicated than merits pursuit here.  However, when the $[\D]_{ii}=0$ entries do not destroy the uniqueness of $\vvh$, it yields a ready result:   

\begin{Corollary}[Lethality Case]
Let all the conditions of Theorem \ref {Theorem:MultivariateMutation} apply except that $[\D]_{ii} = 0$ for at least one $i$.  If $\M_\muv \D$ has a unique eigenvector $\vvh(\muv )$ associated with eigenvalue $\rho(\M_\muv \D)$, then 
\[
\pmu{\rho(\M_\muv \D)}{\kappa} = 0
\]
if and only if
\[
D_{\displaystyle i_1 \cdots i_\kappa \cdots i_L}  = D_{\displaystyle i_1 \cdots i_\kappa' \cdots i_L}
\]
for all $i_\kappa, i_\kappa' \in \{1, \ldots, \nu_\kappa\}$, whenever 
\[
\vh_{i_1 \cdots i_\kappa \cdots i_L} > 0, \mbox{\ \ and \ \ } \vh_{i_1 \cdots i_\kappa' \cdots i_L} > 0.
\]
Otherwise,
\[
\pmu{\rho(\M_\muv \D)}{\kappa} < 0.
\]
\end{Corollary}

%%%%
\begin{proof}
The positivity of $\D$ enters the proof of Theorem \ref{Theorem:MultivariateMutation} only at step \eqref{eq:BKUvvh}.  Assuming the uniqueness of $\vvh$ allows one to preserve \eqref{eq:BKUvvh}.  The first consequence of relaxing the positivity condition does not occur until after \eqref {vvhGTzero} when it can no longer be assumed that $f_{i_1} > 0$.  Continuing with the notation introduced at \eqref {eq:Dform3}, $D_{i_1 i_2} \ f_{i_1} k_{i_2} = \rho \ g_{i_1} \ k_{i_2}$ is satisfied as long as $D_{i_1 i_2} = \rho \ g_{i_1}/ \ f_{i_1}$ for all $i_2$ whenever $f_{i_1} > 0$, that is, whenever $\vh_{i_1 i_2} > 0$, which is what is stated in the corollary using the multilocus notation.

There are further implications for $\D$ (these do not alter the statement of the corollary):  $D_{i_1 i_2} \ f_{i_1} = \rho \ g_{i_1}$, so if some $D_{i_1 i_2} = 0$, then $g_{i_1} = 0$.  But then $D_{i_1 i_2'} f_{i_1} = 0$ for every $i_2'$.  Consequently, either $f_{i_1} = 0$, which means $\vh_{i_1 i_2} > 0$ for every $i_2$, or $D_{i_1 i_2'} = 0$ for all $i_2'$.  Thus, under the condition that $\partial \rho(\M_\muv \D) / \partial \mu_L = 0$, the existence of one lethal haplotype $i_1 i_2$ implies either that all haplotypes with $i_1$ are lethal, or that all haplotypes with $i_1$ are absent from the population.

In the latter case, $f_{i_1} = 0$, there are further implications.  Recall that 
\[
\D (\fv \otimes \kv) = \rho \ \gv \otimes \kv, \mbox{\ and \ } \rho \ \vvh  = \rho \ \fv \otimes \kv = \M_\muv \D  \vvh = \M_\muv \D (\fv \otimes \kv) 
= \rho \ \M_\muv (\gv \otimes \kv) .
\]
So $f_{i_1} = 0$ if and only if $\vh_{i_1 i_2} = f_{i_1} k_{i_2} = 0$, since $k_{i_2} > 0$.  Thus
\eban
\rho \ \vh_{i_1 i_2} = \rho \ f_{i_1} k_{i_2} = \rho \ \sum_{j_1 j_2} M_{i_1 i_2, j_1 j_2} \ g_{j_1} k_{j_2} = 0.
\eean
The zero sum mandates that $g_{j_1} = 0$ for every $j_1$ in which $M_{i_1 i_2, j_1 j_2} > 0$ for some $i_2, j_2$.  By \eqref{eq:Dform3}, $g_{j_1} = 0$ implies $D_{j_1 i_2} \ f_{j_1} = 0$ for all $i_2$, requiring that either $f_{j_1} = 0$, or $D_{j_1 i_2} = 0$ for all $i_2$.

In the case where $f_{j_1}=0$, then the above argument applies in turn to it.  So consider the entire set $\Zc = \{ i_1' \suchthat f_{i_1}' = 0\}$.  

If $M_{i_1 i_2, j_1 j_2} > 0$ only when both $i_1, j_1 \in \Zc$, that means $M_{i_1 i_2, j_1' j_2} = 0$ for all $j_1' \notin \Zc$.  But that means there is no mutation to $\Zc$ from outside of $\Zc$, which makes $\M_\muv$ reducible, contrary to hypothesis.  Therefore there must be some $M_{i_1 i_2, j_1 j_2} > 0$ that has $i_1 \in \Zc$ and $j_1 \notin \Zc$.  And $j_1 \notin \Zc$ implies $D_{j_1 i_2} = 0$ for all $i_2$.  

Therefore, if $\partial \rho(\M_\muv \D) / \partial \mu_L = 0$, the existence of one lethal haplotype $i_1 i_2$ implies either that all haplotypes with $i_1$ are lethal, or that all haplotypes with $i_1$ are absent from the population, which implies further that all $j_1$ in the population that can mutate to $i_1$ are lethal for all haplotypes $j_1 j_2$.  
\end{proof}

%%%%%%%%%%%%%%%%%%%
\begin{Corollary}[Multiple Cell Divisions]
\label{Corollary:MultipleDivisions}
One may substitute $\M_\muv^t$ for $\M_\muv$ in Theorem \ref{Theorem:MultivariateMutation}, where $t$ is a positive integer, and the theorem applies otherwise unchanged.  The partial derivatives, however, are all scaled by $t$.
\end{Corollary}

%%%%
\begin{proof}
The proof is identical to that of Theorem \ref{Theorem:MultivariateMutation} except that one is seeking
\[
\pmu{}{\kappa}  \rho(\M_\muv^t \D)  = \pmu{}{\kappa}  \rho( \left\{ \bigotimes_{\xi=1}^L [ (1-\mu_\xi) \I^{(\xi)} + \mu_\xi \Pm^{(\xi)}  ] \right\}^t \D ).
\]
Following the same sequence of steps as for Theorem \ref{Theorem:MultivariateMutation}, let $\xvh(\muv)$ confined to the unit sphere produce the maximum of 
\[
\phi(\x) \eqdef \xv ^\tp (   \Upsi_\muv^{t/2} \K^\tp \D  \K  \Upsi_\muv^{t/2} ) \xv.
\]
Then $\xvh(\muv)$ is an eigenvector satisfying: 
\eb
 ( \Upsi_\muv^{t/2} \K^\tp \D  \K  \Upsi_\muv^{t/2} ) \ \xvh (\muv)= \rho(\M_\muv^t \D) \ \xvh (\muv).
\ee
Following the same steps of differentiation:
\eban
\pmu{}{\kappa}  \rho(\M_\muv^t \D) 
&=&  2 \xvh(\muv)^\tp ( \Upsi_\muv^{t/2} \K^\tp \D  \K  \Upsi_\muv^{t/2} ) \pmu{\xvh(\muv)}{\kappa} \\
&& +  \ 2 \xvh(\muv)^\tp ({ \Upsi_\muv^{t/2}} \K^\tp \D  \K  \pmu{\Upsi_\muv^{t/2}}{\kappa} ) \xvh(\muv) \\
&=&    t \ \xvh(\muv)^\tp (  \Upsi_\muv^{t/2} \K^\tp \D  \K  \Upsi_\muv^{t/2 - 1} \  \pmu{ \Upsi_\muv}{\kappa} ) \xvh(\muv).
\eean
Utilizing
\eb
\xvh (\muv)^\tp \Upsi_\muv^{t/2} \K^\tp \D  \K   \ = \rho(\M_\muv^t \D) \ \xvh(\muv)^\tp \Upsi_\muv^{-t/2},
\ee
one obtains
\eban
\pmu{}{\kappa}  \rho(\M_\muv^t \D) 
&=&    t \ \rho(\M_\muv^t \D) \  \xvh(\muv)^\tp \Upsi_\muv^{-t/2}  \Upsi_\muv^{t/2 - 1} \  \pmu{ \Upsi_\muv}{\kappa} \  \xvh(\muv) \\
&=&    t \ \rho(\M_\muv^t \D) \  \xvh(\muv)^\tp  (\Upsi_\muv^{- 1} \  \pmu{ \Upsi_\muv}{\kappa} ) \ \xvh(\muv) 
\eean
which is identical to \eqref{eq:pmukXU} except for the presence of $t$.
Since $\Upsi_\muv^{-1}  \ \partial  \Upsi_\muv/ \partial \mu_\kappa$ is negative semi-definite for $\muv \in (0, 1/2)^L$:
\eb
\pmu{}{\kappa}  \rho(\M_\muv^t \D) \leq 0.
\ee

The steps in the equality case are unchanged except for the scaling factor $t$, and the substitution of $t/2$ for $1/2$ in the powers of $\Upsi$.
\end{proof}

%%%%%%%%%%%%%%%%%%%%
\begin{Corollary}[Global Mutation Rate Control]
\label{Corollary:GlobalMutation}
Let all the conditions be identical to those in Theorem \ref{Theorem:MultivariateMutation} except that  the modifier locus controls a single, global mutation rate $\gamma$, scaling all the $\mu_\xi = \gamma \beta_\xi$ parameters equally:
\eb
\M_\muv = \bigotimes_{\xi=1}^L [ (1-\gamma \, \beta_\xi) \I^{(\xi)} + \gamma  \,  \beta_\xi \Pm^{(\xi)}  ].
\ee
Then the spectral radius of 
$
\M_\muv \D
$
is non-increasing in $\gamma $ for $\gamma \in (0, 1/2)$, and strictly decreasing if $\D \neq c \I$ for every $c > 0$.
\end{Corollary}

\begin{proof}
This follows directly from the fact that $\pmu{} {\xi}\rho(\M_\muv \D) \leq 0$ for any $\muv \in (0,1/2)^L$.  If $\D \neq c \ \I$ for any $c > 0$, then for at least one $\xi$, $\pmu{}{\xi}\rho(\M_\muv \D) < 0$, hence  $\dgam{} \rho(\M_\muv \D) < 0$.  This can be shown explicitly.  

Letting $\mu_\xi = \gamma \ \beta_\xi$, we wish to evaluate $\dgam{} \rho(\M_\muv \D) $.  The derivation is identical to the steps in the proof of Theorem \ref{Theorem:MultivariateMutation} except that $d / d \gamma$ replaces $\partial / \partial \mu_\xi$, until step \eqref {eq:pmukUpsi}, which becomes:
\begin{align*}
\dgam{ \Upsi_\muv } 
 = \dgam{} & \left( \bigotimes_{\xi=1}^L [ (1-\gamma \, \beta_\xi) \I^{(\xi)} + \gamma \, \beta_\xi \Lam^{(\xi)}] \right) \nonumber \\
= \sum_{\xi=1}^L & \left\{  \bigotimes_{\xi=1}^{\kappa - 1} [ (1-\gamma \, \beta_\xi) \I^{(\xi)} + \gamma \, \beta_\xi \Lam^{(\xi)}] \nonumber \right. \\
&  \  \otimes \beta_\xi [  \Lam^{(\kappa)} - \I^{(\kappa)}] \finalversion{\ \otimes} \nonumber \\
&  \  \left. \bigotimes_{\xi=\kappa + 1}^{L} [ (1-\gamma \, \beta_\xi) \I^{(\xi)} + \gamma \, \beta_\xi \Lam^{(\xi)}] \right\}, \nonumber \\
= &  \sum_{\xi=1}^L \frac{\beta_\xi}{\gamma} \  \frac{\partial}{\partial \beta_\xi} \Upsi_\muv , 
\end{align*}
Applying this expression yields a positive weighted sum of partial derivatives \eqref {eq:pmukRhoMD}:
\begin{align}
\dgam{}  \rho(\M_\muv \D) 
&=   \xvh(\muv)^\tp (  \Upsi_\muv^{1/2} \K^\tp \D  \K  \Upsi_\muv^{-1/2}  \dgam{ \Upsi_\muv} ) \xvh(\muv) \nonumber \\
&=   \sum_{\xi=1}^L\frac{\beta_\xi}{\gamma} \  \xvh(\muv)^\tp (  \Upsi_\muv^{1/2} \K^\tp \D  \K  \Upsi_\muv^{-1/2} \frac{\partial \Upsi_\muv}{\partial \beta_\xi}   ) \, \xvh(\muv) \nonumber \\
\finalversion{&=  \sum_{\xi=1}^L\frac{\beta_\xi}{\gamma} \frac{\partial }{\partial \beta_\xi}\rho(\M_\muv \D). \label{eq:SumPartials}}
\end{align}
By \eqref{eq:pmukUpsiNegSemiDef}, each of these partial derivative terms \finalversion{in \eqref {eq:SumPartials}} is non-positive, so \linebreak$\dgam{}  \rho(\M_\muv \D)$ $< 0$ if at least one term is non-zero.  To have all partial derivatives be $0$ requires $\D = c \ \I$ for some $c > 0$, hence  $\dgam{}  \rho(\M_\muv \D) < 0$ if $\D \neq c \ \I$ for every $c > 0$.  \finalversion{Note that if the modifier scales the mutation rates of only a subset of loci, then the sum in \eqref {eq:SumPartials} is replaced with a sum over that subset of loci.  Hence the magnitude of $\dgam{}  \rho(\M_\muv \D) $ increases with the number of loci affected by the modifier, given fixed $\beta_\xi$ values.}
\end{proof}

%%%%%%%%%%%%%%%%%%%%%%%%%%%%%%%%%%%%%
\subsection{Neutral Surfaces of Mutation Rates} \label{subsection:NeutralSurfaces}

Theorem \ref {Theorem:MultivariateMutation} shows that if the marginal fitnesses at equilibrium do not depend on the allelic state of a particular locus (i.e. no instance of \eqref {eq:Differ} occurs), then the mutation rate for that locus can be varied without changing the spectral radius of the stability matrix.  This trivially defines a surface of points $\etav \in (0, 1/2)^L$ on which $\rho(\M_\etav \D)$ is invariant.  Let us exclude this degenerate case for this section, and assume that the marginal fitnesses at equilibrium depend on every locus.  

Under this assumption, $\rho(\M_\muv \D)$ strictly decreases in each variable $\mu_\kappa$.  This raises the question of how $\rho(\M_\muv \D)$ and $\rho(\M_\etav \D)$ compare for two vectors $\muv$ and $\etav$ when $\mu_\kappa < \eta_\kappa$ for some $\kappa$, but $\mu_\xi > \eta_\xi$ for some other $\xi$, i.e. $\muv$ and $\etav$ are not ordered componentwise, and neither $\muv \leq \etav$, nor $\muv \geq \etav$.  

The Intermediate Value Theorem \citep[p. 154]{Munkres:1975} tells us that there must be a set, $\Nc (\muv) \subset \Reals^n$, surrounding $\muv$, on which $\rho(\M_\etav \D) = \rho(\M_\muv \D)$ for all $\etav \in \Nc (\muv)$; this is because $\rho(\M_\muv \D)$ is a continuous function from the matrix entries of $\M_\muv \D$ to $\Reals$ \citep[pp. 539--540]{Horn:and:Johnson:1985}, and the entries of $\M_\muv$ are continuous functions of each $\mu_\kappa$.  
The following properties will be shown for this set $\Nc(\muv)$:  
\benu
\item $\Nc (\muv) $ passes through every orthant surrounding $\muv$ except the strictly positive and strictly negative orthants; 
\item $\Nc (\muv) $ disconnects the mutation parameter space $(0,1/2)^L$ into two connected parts; and 
\item $\Nc (\muv) $ is an $L-1$ dimensional smooth manifold.  
\eenu
In a series of lemmas, the first two properties are established for arbitrary continuous, strictly decreasing functions, and the third is established for arbitrary differentiable functions with negative partial derivatives.  These lemmas are then applied to the mutation rate model, in Theorem \ref{Theorem:MutationManifold}.

\begin{Lemma}[Orthants]
\label{Lemma:Orthants}
Let $F \suchthat \Reals^L \rightarrow \Reals$ be continuous and strictly decreasing in each variable $m_i$ of $\m \in \Reals^L$. 

Let the orthants of $\Reals^L$ be represented as follows:  $({\Ic^+, \Ic^=, \Ic^-})$ represents a three-way partition of the indices $i = 1 \ldots L$.  The orthant $Q({\Ic^+, \Ic^=, \Ic^-}) \subset \Reals^L$ is defined as:
\eb
Q({\Ic^+, \Ic^=, \Ic^-}) = \{ \m \in \Reals^L \suchthat 
\left\{ 
\begin{array}{lr}
m_i > 0 & \forall \ i \in \Ic^+\\
m_i = 0 & \forall \ i \in \Ic^=\\
m_i < 0 & \forall \ i \in \Ic^- 
\end{array} 
\right.
\}.
\ee

Then, for any non-empty choices of subsets $\Ic^+$ and $\Ic^-$,  there is some \\
$\qv \in Q({\Ic^+, \Ic^=, \Ic^-})$ such that $F(\m + \qv) = F(\m)$.
\end{Lemma}
\begin{proof}
Choose an arbitrary $\qv \in Q({\Ic^+, \Ic^=, \Ic^-})$.  If $F(\m + \qv) = F(\m)$, then one has found the sought-after $\qv$.  

If $F(\m + \qv) > F(\m)$, then one can construct a $\qv'$ such that $F(\m + \qv') < F(\m) < F(\m + \qv)$, and then find the desired value between $\m + \qv$ and $\m + \qv'$:  increase all the negative elements of $\qv$ to 0 to define $\qv' \suchthat q_i' = 0$ for all $i \in \Ic^-$,  $q_i' = q_i$ for $i \in \Ic^+, \Ic^=$.  Thus $\qv' \in Q(\Ic^+, \Ic^= \cup \Ic^-, \emptyset)$.  By the monotonicity of $F$ one knows $F(\m + \qv') < F(\m + \qv)$.  And $\qv' \in Q(\Ic^+, \Ic^= \cup \Ic^-, \emptyset)$ means $\qv' \geq \neq \Zero$, so $F(\m + \qv') < F(\m) < F(\m + \qv)$.  

Now consider the convex combination $\qv(\alpha) \eqdef (1-\alpha) \qv + \alpha \qv'$.  Since $F$ is a continuous function, with $F(\m + \qv(0)) = F(\m + \qv) > F(\m) >  F(\m + \qv') = F(\m + \qv(1)) $, then by the Intermediate Value Theorem \citep[p. 154]{Munkres:1975}, there is some $\alphah \in (0, 1)$ such that $F(\m + \qv(\alphah) ) = F(\m)$.  We must verify that $\qv(\alphah) \in Q({\Ic^+, \Ic^=, \Ic^-})$ as required:   for $i \in \Ic^-$, $q(\alphah)_i = (1-\alphah) q_i + \alphah q'_i =  (1-\alphah) q_i < 0$.  So $\qv(\alphah)$ is the desired point.

If $F(\m + \qv) < F(\m)$, the mirror argument applies, and we decrease the $q_i \suchthat i \in \Ic^+$ to make a new point with $q'_i = 0$ giving $F(\m + \qv') > F(\m)$ where $\qv' \in Q(\emptyset, \Ic^+ \cup \Ic^=, \Ic^-)$.  Analogously, we know there is $\alphah$ that yields $F(\m + (1-\alphah)\qv+ \alphah \qv') ) = F(\m)$.
\end{proof}

%%%%%%%%
\begin{Lemma}[Connected Regions]
\label{Lemma:ConnectedRegions}
Let $F \suchthat (0,c)^L \rightarrow \Reals$ be continuous and strictly decreasing in each variable $m_i$ of $\m \in (0, c)^L \subset \Reals^L$, where $c > 0$ is a constant.  Then the set $\Nc(\m) = \{\m' \colon F(\m') = F(\m) \}$ disconnects $(0, c)^L$ into two connected sets.
\end{Lemma}
%%%%
\begin{proof}
To show that the set $\Nc(\m)$ disconnects $(0,c)^L$, it is sufficient to find two points in $(0,c)^L$ such that every continuous curve between them intersects $\Nc(\m)$.
For the given $\m \in (0,c)^L$, let $q_i \eqdef \min ( m_i, c-m_i) /2 > 0$.  The two requisite points will be  $\m - \qv $, and $\m + \qv$.  Clearly $0 < m_i - q_i < m_i < m_i + q_i < c$, so $\m- \qv, \m + \qv \in (0,c)^L$.  Since $F$ decreases in each variable, then $F(\m - \qv) > F(\m) > F(\m + \qv)$.  Define the continuous curve $C \suchthat [0,1] \mapsto (0,c)^L$ to have $C(0) = \m - \qv $, and $C(1) = \m + \qv$.  Since $F$ is continuous and $C$ is continuous, then $F \circ C \suchthat [0, 1] \rightarrow \Reals$ is continuous.  Since 
\[
F(C(0)) = F(\m - \qv) > F(\m) > F(C(1)) = F(\m + \qv),
\]
by the Intermediate Value Theorem there must be some $\alpha \in [0,1]$ such that $F(C(\alpha)) = F(\m)$, which means $C(\alpha) \in \Nc(\m)$.  Thus every continuous curve between $\m - \qv$ and $\m + \qv$ intersects $\Nc (\m) $.  Therefore $\Nc (\m) $ disconnects $(0,c)^L$.

To show that $(0,c)^L - \Nc(\m)$ consists of two connected sets, it is sufficient to show that between any pair of points in the same set, there is a continuous curve $\subset (0,c)^L$, that does not intersect $\Nc(\m)$.   
The two sets are
$$\Sc^-(\m) \eqdef \{ \pv \in (0,c)^L \suchthat F(\pv) < F(\m) \}$$ 
and 
$$\Sc^+(\m) \eqdef \{ \pv\in (0,c)^L \suchthat  F(\pv) > F(\m) \}.$$  
Clearly $\Sc^-(\m)$, $\Nc(\m)$, and $\Sc^+(\m) $ are disjoint, and $\Sc^-(\m) \union \Nc(\m) \union \Sc^+(\m) = (0,c)^L$.  Now we shall see that $\Sc^-(\m) $ is connected, and $\Sc^+ (\m) $ is connected.

For any two distinct points $\pv, \pv' \in \Sc^-(\m) $, construct a new point,  $\pv''$, that combines the maxima from the two points: $p''_i = \max (p_i, p_i')$.  Let one curve be the line from $\pv$ to $\pv''$, $\{ (1-\alpha) \pv + \alpha \pv'' \suchthat \alpha \in [0, 1] \}$, and the other curve be the line from $\pv'$ to $\pv''$, $\{ (1-\alpha) \pv' + \alpha \pv'' \suchthat \alpha \in [0, 1] \}$.  Both curves are clearly within $(0, c)^L$.  The union of the two curves forms a continuous curve between $\pv$ and $\pv'$.  

Now it must be verified that the curves do not intersect $\Nc(\m)$.  For the line between $\pv$ and $\pv''$:  $p''_i \geq p_i$ for each $i$, so $ (1-\alpha) p_i + \alpha p''_i \geq p_i$ for all $\alpha \in [0, 1]$.  Therefore $F( (1-\alpha) \pv + \alpha \pv'') \leq F(\pv) < F(\m)$ for all $\alpha \in [0, 1]$.  So all the points on the line remain within $\Sc^-$.  The same applies to the other line between $\pv'$ and $\pv''$.  Thus there is a continuous curve connecting $\pv$ and $\pv' \in \Sc^-(\m) $ that does not intersect $\Nc(\m)$.  Therefore $\Sc^-(\m)$ is connected.  The mirror argument applies to $\Sc^+ (\m) $.
\end{proof}

\begin{Lemma}[Smooth Manifold]
\label{Lemma:Manifold}
Let $F \suchthat (0,c)^L \rightarrow \Reals$ be a smooth map, and strictly decreasing in each variable $m_i$ of $\m \in (0, c)^L \subset \Reals^L$, where $c > 0$ is a constant, and $L \geq 2$.  Further, let $\partial F(\m) / \partial m_i < 0$ for each $i$.  Then the set $\Nc(\m) = \{\m' \colon F(\m') = F(\m) \}$ is a smooth, $L-1$ dimensional submanifold of $ (0,c)^L$.
\end{Lemma}
\begin{proof}
The proof is immediate using a general form of the Implicit Function Theorem \citep[p. 135]{Singer:and:Thorpe:1967}), referred to as the Preimage Theorem in \citet[p. 21]{Guillemin:and:Pollack:1974}, or the Regular Value Theorem \citep[Theorem 3.3 p. 22]{Hirsch:1976}, which I restate:
%%%%%%%%
\begin{Theorem}[Implicit Function] \protect{\citep[p. 135]{Singer:and:Thorpe:1967}}.
Let $X$ and $Y$ be smooth manifolds, with $\dim X > \dim Y$.  Let $\psi: X \rightarrow Y$ be a smooth map.  Let $y_0 \in \psi(X)$ and let
\[
X_0 = \psi^{-1}(y_0) = [x \in X \colon \psi(x) = y_0].
\]
Assume that for each $x \in X_0$, $d \psi(x) \colon T(X,x) \rightarrow T(Y, \psi(x))$ is surjective, i.e. the $\dim X \times \dim Y$ matrix
\[
\matrx{ (\partial / \partial x_j) (y_i \circ \psi)|_{x_0} }
\]
is full rank, $\dim Y$.  Then $X_0$ has a manifold structure, whose underlying topology is the relative topology of $X_0$ in $X$, and in which the inclusion map $X_0 \rightarrow X$ is smooth.  Furthermore, $\dim X_0 = \dim X - \dim Y$.
\end{Theorem}
%%%%%%%%%%%%%%%%%

Here, let $X = (0, c)^L$, $Y = \Reals$, $\psi = F$, $x_0 = \m$, $\dim X = L$, $\dim Y = 1$, $X_0 = \Nc(\m)$, and
\[
d \psi(x) = \matrx{ (\partial / \partial x_j) (y_i \circ \psi)|_{x_0} } = \matrx{\partial F(\m) / \partial m_i |_{\m} }.
\]
Here, $d \psi(x)$ is surjective if $\partial F(\m) / \partial m_i \neq 0$ for at least one $i$.  In fact, by hypothesis $\partial F(\m) / \partial m_i < 0$ for every $\m$ and $i$ (so every value $F(\m)$ is a regular point,  making $F$ a submersion).  Thus, $\Nc(\m)$ is a smooth submanifold of $(0, c)^L$ with $\dim \Nc(\m) = \dim (0, c)^L - \dim \Reals = L - 1$.
\end{proof}

These lemmas are now applied to the modifier model:

\begin{Theorem}[Manifold of Neutral Mutation Rates]
\label{Theorem:MutationManifold}
Assume the conditions of Theorem \ref{Theorem:MultivariateMutation}.  For any given mutation rate vector $\muv \in (0, 1/2)^L$, the set of mutation rate vectors that produce the same spectral radius as $\muv$, $\Nc(\muv) = \{\etav \in (0, 1/2)^L \suchthat \rho(\M_\etav \D) = \rho(\M_\muv \D) \}$, has the following properties:
\benu
\item \label{item:Orthants} There is some $\etav \in \Nc(\muv)$ in every orthant around $\muv$ except the orthants $\etav \leq \neq \muv$, and $\etav \geq \neq \muv$;
\item \label{item:Disconnects} $\Nc(\muv)$ disconnects the mutation parameter space $(0, 1/2)^L$ into two connected parts, $\Sc^-(\muv) $ and $\Sc^+ (\muv) $, such that $\rho(\M_\etav \D) < \rho(\M_\muv \D)$ for all $\etav \in \Sc^-(\muv) $, and $\rho(\M_\etav \D)  > \rho(\M_\muv \D)$ for all $\etav \in \Sc^+ (\muv) $.
\item \label{item:Manifold} $\Nc(\muv)$ is a smooth manifold of dimension $L-1$, which is a subset of an affine algebraic variety. 
\eenu
\end{Theorem}
%%%%
\begin{proof}
Let $F(\muv) \eqdef \rho(\M_\muv \D)$.  
From Theorem \ref{Theorem:MultivariateMutation}, $\rho(\M_\muv \D)$ is continuous and strictly decreasing in each mutation rate $\mu_\kappa$.  This satisfies the conditions of Lemma \ref{Lemma:Orthants}, and establishes \ref {item:Orthants}.  Further, with $c = 1/2$ the conditions of Lemma \ref {Lemma:ConnectedRegions} are satisfied and \itemref {item:Disconnects} established.  As Lemma \ref{Lemma:Manifold} requires, $(0, 1/2)^L$ and $\Reals$ are smooth manifolds, and $\rho(\M_\muv \D)$ is a smooth map with respect to $\muv$ when $\M_\muv \D$ is irreducible as it is in Theorem \ref{Theorem:MultivariateMutation} (since for simple eigenvalues, all orders of partial derivatives with respect to the matrix entries exist \citep[p. 2]{Deutsch:and:Neumann:1984}) .  The last requirement of Lemma \ref{Lemma:Manifold} is met since Theorem \ref{Theorem:MultivariateMutation} shows $\partial \rho(\M_\muv \D) / \partial \mu_\kappa < 0$ for each $i = 1$ to $L$, therefore \itemref{item:Manifold} is established.

It can be seen that $\Nc(\muv)$ is a subset of an affine algebraic variety, because $\Nc(\muv) \subset (0, 1/2)^L \intersect \Vc $, where $\Vc$ is the affine variety
$$
\Vc = \{ \etav \suchthat \det[ \bigotimes_{\xi=1}^L [ (1-\eta_\xi) \I^{(\xi)} + \eta_\xi \Pm^{(\xi)}  ] \D -  \rho(\M_\muv \D) \ \I] = 0 \}. \qedhere
$$
\end{proof}

%%%%%%%%%%%%%%%%%%%%%%%%%%%%%%%
\subsection{Main Results}
\label{sec:Main}

Theorem \ref{Theorem:MultivariateMutation}, Corollary \ref{Corollary:MultipleDivisions}, and Theorem \ref{Theorem:MutationManifold} may now be applied to the dynamics of the modifier gene model:

\begin{Theorem}[Multivariate Reduction Principle for Symmetrizable Mutation Rates at Multiple Loci]
\label{Theorem:Main}
Consider a genetic system in which a modifier locus controls the mutation rates of a group of loci under viability selection.  Mutations occur independently among the loci under selection.  In a population near equilibrium under a stable mutation-selection balance, fixed at the modifier locus, let a new allele of the modifier locus be introduced.  The new modifier allele can change the mutation rate parameter separately for each locus, and each parameter scales equally the probability of mutations at that locus. 

Under the following constraints:
\benu
\item mutation rates at each locus range between 0 and $1/2$,
\item no recombination or other transformation process acts on the genes,
\item the mutation matrix for each locus is irreducible, and also irreducible when restricted to nonlethal alleles,
\item is the transition matrix for some reversible Markov chain,
\eenu
then the new modifier allele will increase (decrease) in frequency at a geometric rate if, among the loci that affect the marginal fitnesses of the haplotypes present in the population:
\begin{enumerate}
\item  it reduces (increases) the mutation rate at any locus, and does not increase (decrease) the mutation rates at any locus;
\item it increases the mutation rates for at least one locus, and decreases the mutation rates for at least one locus, and falls below (above) the neutral manifold of mutation rates that includes the mutation rates at the equilibrium.   Should the mutation rates produced by the new modifier allele fall on this neutral manifold, then it will not change frequency at a geometric rate.
\end{enumerate}
Moreover, the further that the new set of mutation rates is from the neutral manifold, the stronger is the eventual induced selection for (against) the new modifier allele, up to a maximum fitness of $\max_i \wh_i / \wbh$ for a modifier allele that eliminates all mutation.

These results hold, in the case of multicellular organisms, for arbitrary numbers of cell divisions between gamete generations.  \finalversion{The strength of selection on the modifier locus scales in proportion to the number of cell divisions in the germline, and increases with the number of loci controlled by the modifier.}
\end{Theorem}
%%%%
\begin{proof}
In the single cell-division model, the population begins at equilibrium fixed on modifier allele $b$ which yields mutation rate vector $\muv$, and \eqref{eq:EquilibriumFixed} becomes:
\[
\zvh_{b} =   \M_\muv \; \D  \; \zvh_{b},
\]
Therefore, in \eqref{eq:BKUvvh}, $\vvh(\muv) = \zvh_b$, and $\rho(  \M_\muv \D ) = 1$.\footnote{The use of this equilibrium relation, without having to explicitly solve for the equilibrium, was first introduced into modifier gene theory by \citet[p. 89]{Teague:1977}.}  Let $\etav$ be the vector of mutation rates produced by the new modifier allele $a$.  If $\etav \leq \muv$ and $\eta_\kappa < \mu_\kappa$ for some locus $\kappa$ for which the equilibrium marginal fitnesses depend on the alleles at locus $\kappa$, then by Theorem \ref{Theorem:MultivariateMutation}, $\rho(\M_\etav \D) > \rho(\M_\muv \D) = 1$, so new modifier allele $a$ increases at a geometric rate.  The mirror argument applies when $\etav \geq \muv$ and $\eta_\kappa > \mu_\kappa$ for some locus $\kappa$ for which the equilibrium marginal fitnesses depend on the alleles at locus $\kappa $, in which case the new modifier allele will decrease at a geometric rate.

In the case where $\eta_\kappa > \mu_\kappa$ and $\eta_j < \mu_j$ for some $\kappa \neq j$, Theorem \ref {Theorem:MutationManifold} establishes that there is a smooth $L-1$ dimensional surface $\Nc(\muv)$ that dissects this orthant surrounding $\muv$ into a set below $\Nc(\muv)$ in which $\rho(\M_\etav \D) > 1 $, and a set above $\Nc(\muv)$ in which $\rho(\M_\etav \D) < 1$, the new modifier allele increasing in frequency in the former case, and decreasing in the latter case.

If $\etav \in \Nc(\muv)$, then by definition $\rho(\M_\etav \D) = 1$, so the new allele will not change frequency at a geometric rate.

By `further from' the neutral manifold, I mean a partial ordering of mutation rate vectors in which $\muv_1 \prec \muv_2$ if $\muv_1 - \muv_2 \leq \neq \0$ (not equal for at least one  locus that the equilibrium fitnesses $\wh_i$ depend on).  Since the derivative of $\rho(\M_\etav \D)$ is negative with respect to each variable $\eta_\xi$ when locus $\xi$ affects $\wh_i$, if $\muv_1 \prec \muv_2 \prec \muv_3$, then $\rho(\M_{\muv_1} \D) > \rho(\M_{\muv_2} \D) > \rho(\M_{\muv_3} \D)$.  For a modifier allele that eliminates mutation, $\muv = \0$, so $\rho(\M_\0 \D) = \rho(\D) = 1 + V = \max_i \wh_i / \wbh$.

In the multiple cell-division model, the initial equilibrium satisfies:
\[
\zvh_{b} =   \M_\muv^t \; \D  \; \zvh_{b},
\]
so $\rho(\M_\muv^t  \D) = 1$.  From Corollary \ref{Corollary:MultipleDivisions}, we see that letting $t \geq 1$ does not alter the inequalities on the spectral radius, so the same conclusions apply for all $t$.
\end{proof}

\section{Discussion}

The motivation for the paper was to extend the general theory of modifier genes beyond single event models and the constraint of linear variation.  Here, multiple independent mutations among multiple loci are modeled, with a modifier gene that has arbitrary control of the mutation rates at each locus.  Under this multivariate, multiplicative form of variation, the reduction principle is again found to hold.  In particular:  
\benu
\item  The result applies for arbitrary selection coefficients on the diploid genotypes (with some technical constraints on the global pattern of any lethal genotypes), arbitrary mutation rates and mutation distributions and as long as they are symmetrizable, arbitrary numbers of (tightly linked) loci and alleles, arbitrary control over each single-locus mutation rate, and any number of cell divisions in the germline.  
\item \label{item:DisI} Changes in the mutation rate at a locus will be neutral if the alleles at that locus do not make any difference in the marginal fitnesses of the haplotypes under selection. 
\item There is a surface of mutation rates that a new modifier allele can produce that leave it neutral, i.e.  it will not change frequency at any geometric rate when introduced into the population.
\item Mutation rates that fall below this surface will cause the new modifier allele to increase when rare, and rates above this surface will cause it to go extinct.  The surface is such that the modifier allele can increase the mutation rate at some loci and decrease at others --- for any arbitrary choice of loci that affect the marginal fitnesses at equilibrium --- and there will always be some values for the magnitude of these changes that fall below the neutral surface of mutation rates, and other values that fall above.
\item The strength of selection on a new modifier allele increases with the distance of its rates from the neutral surface of mutation rates, \finalversion{which increases with each locus affected, and it increases} with the number of cell divisions in the germline.
\eenu

Two properties of modifier polymorphisms are also shown:  
\benu  
\item The ``viability analogous, Hardy-Weinberg'' modifier polymorphisms that emerge in single-event models cannot exist under multivariate, multiplicative variation in transmission due to the loss of convexity in of the space of transmission values.  
\item When the modifier locus is polymorphic, the only values that matter to the change in frequencies of the loci under selection are the mean transmission probabilities for those loci; the frequencies and associations of the modifier alleles are otherwise irrelevant.
\eenu

%%%%%%%%%%%%%%%%%%%%%%%%%%%%%
\subsection{Q \& A}

Since the implications of the main results may not be immediately apparent, an attempt to elucidate them is provided through the following `Question and Answer' format.

\newcounter{QA}
\begin{list}
{\bf Q.\arabic{QA}.}{\usecounter{QA}}

\item \label{Q1} What new phenomena are found in these results? 

{\bf A.}  While the general result that mutation rates evolve to decrease is not novel, several phenomena are:  \benu
\item \label{item:Compensate} Increases in mutation rates may evolve if they are compensated for by decreases at other loci (see \Secref{subsec:Multivariate} below).  
\finalversion{\item The strength of selection for (against) a new modifier grows with the number of loci whose mutation rates it decreases (increases) (see Corollary \ref{Corollary:GlobalMutation}).}
\item  \label{item:Impunity} Mutation rates of loci that do not affect the marginal fitnesses at equilibrium may be changed `with impunity' by the new modifier allele, including when they are changed as a side effect of changes in the mutation rates at other loci.  This implies --- other things being equal --- that if there is local tuning of mutation rates, then neutral loci should have greater mutation rates that loci held in mutation-selection balance (see \Secref {subsubsec:Correlations}).
\item The reduction principle applies when there are multiple cell divisions in the lineage from zygote to gamete, \finalversion{and the strength of selection on the modifier locus scales in proportion to the number of cell divisions from zygote to gamete (Corollary \ref{Corollary:MultipleDivisions})}.
\eenu

\item \label{Q2 }Why is this called `evolutionary reduction' if it is possible for some mutation rates to evolve an increase?

{\bf A.}  It is a `reduction result' because mutation rates must be below the neutral manifold in order for the new modifier allele to invade --- making the neutral manifold like a wall (see \Secref{subsubsec:Manifold}).  Also, the further below this wall that the mutation rates are, the stronger the induced selection for the modifier allele carrying \finalversion{them.}

\item \label{Q3} Could there be some sort of complex epistatic multi-locus selection regime that would allow mutation rates to get around this wall?

{\bf A.}  No, the neutral manifold emerges for all possible selection regimes, the only `holes' in the wall being the mutation rates of loci that do not affect the marginal fitnesses at equilibrium, which are free to evolve in any direction.

\item \label{Q3A}  Doesn't the reduction result depend on the assumption that most mutations are deleterious?  Why isn't this assumption stated anywhere in the model?

{\bf A.}  The reduction result does not depend on any assumption that most mutations are deleterious --- and that is why it takes some mathematical machinery to show it.  What it does depend on is a \emph{net flux} of mutations at equilibrium from more fit to less fit haplotypes, which is a necessary and emergent property of mutation-selection balances.  By net flux I mean an absence of the `detailed balance' condition that characterizes the \finalversion{stationary} state of reversible Markov chains, \eqref{eq:ReversibleCriterion}, in which the fraction of the population mutating from type $j$ to $k$ equals the fraction mutation from $k$ to $j$.  At a mutation-selection balance, a net flux is necessary to keep the haplotypes with above average fitness from continuing to grow in frequency, and to keep the haplotypes with below average fitness from continuing to decline in frequency.  This outcome will occur regardless of how the \emph{distribution of fitness effects} (DFE) \citep{Eyre-Walker:and:Keightley:2007} is set for each diploid genotype by nature.

By altering the flux, the new modifier allele unbalances the mutation-selection balance within the subpopulation that contains it.  It is not immediately obvious why a reduction in the flux equally across all haplotypes (linear variation) would create a subpopulation with increased mean fitness, because the flow is reduced in both directions:  from less fit to more fit, and from more fit to less fit.  But the net effect is always to increase the subpopulation's mean fitness, as shown by Theorem 5.2 \citet{Karlin:1982}.  Here, fluxes are scaled equally between all single locus alleles, multiplied across loci, and Theorem \ref{Theorem:MultivariateMutation} gives the multivariate reduction result.

\item \label{Q4}   In nature, are not the rates of mutations that affect the phenotype so low that multiple mutations in a gamete are very rare? --- in which case, don't the results here reduce to the classical results for single events?

{\bf A.}  No, for several reasons:
\benu
\item Phenomena \itemref{item:Compensate} and  \itemref{item:Impunity} in \Qref{Q1} above, are novel to the multivariate control of mutation rates, and are not eliminated in the limit of small mutation rates.  In this limit, when multiple mutations are rare enough to be ignored,
\begin{align}
\label{eq:SmallMu}
\M_\muv &= \bigotimes_{\xi=1}^L [ (1-\mu_\xi) \I^{(\xi)} + \mu_\xi \Pm^{(\xi)}  ]  \\
&= \I + \sum_{\kappa=1}^L  \mu_ \kappa \left\{\bigotimes_{\xi =1}^{\kappa - 1}  \I^{(\xi)} 
 \otimes [  \Pm ^{(\kappa)} - \I^{(\kappa)}] \finalversion{\otimes} \bigotimes_{\xi=\kappa + 1}^{L}   \I^{(\xi)}  \right\}+ \Order{(\mu_\xi^2)}\notag
\end{align}
Ignoring the $\Order{(\mu_\xi^2)}$ terms, the space of variation $\Mc = \{ \M_\muv \}$ becomes convex, and the ``viability analogous, Hardy-Weinberg'' equilibria become feasible.  A modifier allele that scales all $\mu_\xi$ equally will produce linear variation, as is covered by earlier treatments \citep{Altenberg:1984,Altenberg:and:Feldman:1987,Altenberg:2009:Linear}.  But modifiers that change the ratios between the $\mu_\xi$ do not produce linear variation, and are subject to phenomena \itemref{item:Compensate} and  \itemref{item:Impunity} in \Qref{Q1} above.
 
\item \label{QA:Poisson} Mutation rates observed in organisms are actually not small enough to ignore multiple mutations.   For example, \citet{Roach:Glusman:etal:2010} estimate that humans have some 70 new nucleotide mutations per diploid genome per generation.  On a per-cell division basis, this puts the human germline mutation rate lower than that recorded for any other species \citep{Lynch:2010:Rate}.  For the fraction of these mutations that have phenotypic effect, \citet{Eyre-Walker:and:Keightley:2007} summarize several studies that estimate the proportion of the genome subject to natural selection at around 5\% in mammals.  

Letting $\lambda$ be the number of non-neutral mutations per haplotype per generation, this yields and estimate of $\lambda = 0.05 \times 70 / 2 = 1.75$.  \citet{Lynch:2010:Rate} gives a concordant estimate of 0.9 to 4.5 deleterious mutations per diploid genome per generation, or $0.45 \leq \lambda \leq 2.25$ per haplotype.
 
With this magnitude for the expected number of mutations, a modifier allele that changes the global mutation rate will not be producing linear or even convex variation.  When $\lambda = \mu \ L << L$, the multiplicative model is approximated by a Poisson process, where the probability of parent $j$ producing gamete $i$ with $\nu $ mutations is $\lambda^\nu  e^{-\lambda} / \nu !$.  The ratio between gametes with multiple mutations and gametes with single mutations is:
\begin{align*}
\frac{\Pr[\geq 2]}{\Pr[1]} &=  \frac{ 1 - (1+ \lambda) e^{-\lambda} }{\lambda e^{-\lambda}}
=  \frac{1 }{\lambda} ( e^{\lambda} - 1) -   1 
=  \lambda \sum_{\nu =0}^\infty \frac{\lambda^{\nu } }{(\nu +2)!} 
\end{align*}
At the small mutation rate limit,
$$
\lim_{\lambda \rightarrow 0} \frac{\Pr[\geq 2]}{\Pr[1]}  = 0,
$$
but for $\lambda = 1.75$, 
\begin{align*}
\frac{\Pr[\geq 2]}{\Pr[1]} &=  \frac{1 }{1.75} ( e^{1.75} - 1) -   1 = 1.72 >> 0.
\end{align*}
Here there are about twice as many multiple mutations as single mutations.  The range $0.45 \leq \lambda \leq 2.25$ gives $0.26 \leq \Pr[\geq 2] / \Pr[1]  \leq 2.8$.  So multiple mutations cannot be ignored in \eqref{eq:SmallMu}.   Other eukaryotic species whose mutation rates have been measured give rates of deleterious mutations of $\lambda > 0.5$ \citep[p. 1171]{Kondrashov2:2010}.  Therefore, when the expected number of mutations is on the order of 1, the multiplicative model is not approximated by the classical model of linear variation and requires an approach such as taken here.

\item Lastly, a full picture of the evolution of mutation rates must take into account not only the `wild type' mutation rates --- which are the endpoint of the evolutionary process --- but also the full range of mutation rates that organisms are capable of generating, because they are the values that test the  evolutionary stability of the wild-type values.   In humans, somatic cells exhibit mutation rates that are one to two orders of magnitude greater than germline cells \citep{Lynch:2010:Rate}.  This shows that  human cells are capable of producing many mutations per generations, and makes necessary a treatment  that can handle multiple mutations in order to analyze the evolutionary stability of low mutation rates.
\eenu
%%%%%%%%%%

\item \label{Q5}  This result holds for populations fixed at the modifier locus.  What can we expect if the initial population is polymorphic for the modifier?

{\bf A.}  Here the analytical techniques break down, because there is no clear relationship between the variation in transmission produced by a new modifier allele and the mean transmission probabilities \eqref {eq:Tbar} in the population.  Based on the ubiquity of the reduction result, one can conjecture that a form of reduction result will hold, but its exact form requires analysis that can handle more general forms of variation in transmission.

\item \label{Q6} Once a new modifier allele successfully invades, what happens then?

{\bf A.}  The results here are for local perturbations of the equilibrium population, and so do not reveal what happens once a modifier allele invades.  As the new modifier allele increases in frequency, it obviously changes the mutation rates experienced by the loci under selection.  If these changes are small enough, then $\bar{T}_{(r)}(i \la j | k)$ in \eqref {eq:Tbar} will change only slightly, and by the `theory of small parameters' \citep{Karlin:and:McGregor:1972:Application}, the haplotype frequencies of loci under selection will converge to another stable equilibrium near the starting stable equilibrium.  

For modifiers with larger effects, however, the original equilibrium can potentially become unstable or even disappear.  Homotopy continuation methods may be of use in elucidating the possibilities here.  

Whatever the population re-equilibrates to after invasion of the modifier allele, it is again subject to invasion by additional modifier alleles that reduce the mutation rates below the current  neutral manifold of mutation rates.

\item \label{Q7} What guess can be made as to how the inclusion of recombination would change the results?

{\bf A.}  The inclusion of mutation makes the model into an example of a `mixed process', which is where departures from the reduction result have been found.  \citet{Holsinger:and:Feldman:1983:MM} find that maximal mutation rates evolve in a model of pure selfing and overdominance --- which is a mixed process --- because selfing drives down the frequency of the fittest genotype, the heterozygote, which high mutation helps to restore.  So, could mutation restore the frequencies of high fitness genotypes that recombination drives down?  Such a situation is difficult to imagine, because the genotypes that recombination would drive down are overdominant coadapted gene complexes, and it seems unlikely that mutation would help to boost the frequency of such complexes.

As to recombination between the modifier gene and the loci under selection, it in essence dilutes the subpopulation by mixing in some of the equilibrium population.  So recombination would expected be to moderate the force of selection induced on a new modifier allele, but not to change its direction.

\item \label{Q8} What do these results have to say about populations not at equilibrium?

{\bf A.}  Almost, but not quite, nothing.  Populations that are far from equilibrium --- due to small populations, populations under varying selection, populations in transient phases of evolution, and populations evolving with novel genotypes --- have modifier gene dynamics that are fundamentally different from the equilibrium populations considered here.  However, at some point where the population becomes `close enough' to equilibrium, the near-equilibrium dynamics will again take hold.  This appears to be seen, for example, by \citet{Giraud:Matic:etal:2001} in enteric bacterial populations, which are far from equilibrium when first colonizing a new host, and evolve higher mutation rates, but after some period of time evolved reduced mutation rates.  So at some point with large enough population sizes, slow enough variation in selection, damped out transients, or rare enough novel genotypes, the results for near-equilibrium models should come to dominate the dynamics.  
\end{list}

%%%%%%%%%%%%%%%%%%%%
Next, details of additional aspects of the results will be discussed:  The constraint that mutation be symmetrizable, the multivariate reduction principle, the strength of selection on the modifier locus, and models that depart from the reduction principle.

%%%%%%%%%%%%%%%%%%%%%%%%%%%%%
\subsection{The Symmetrizable Mutation Constraint}

The constraint that the mutation matrices be symmetrizable is necessary to use the Rayleigh-Ritz variational characterization for the spectral radius.  It causes all the eigenvalues of $\M_\muv$ and $\M_\muv \D$ to be real.  Since symmetrizable $\M$ is the transition matrix for a reversible Markov chain, its Perron vector produces `detailed balance' \eqref{eq:ReversibleCriterion}.  The mathematical tractability of reversible Markov chains has led to their widespread used in phylogenetic inference models, regardless of whether empirical mutation rates actually are symmetrizable (\citealt{Rodriguez:et:al:1990,Yang:1995,Jayaswal:Jermiin:and:Robinson:2005, Squartini:and:Arndt:2008}).

There is no reason to believe, however, that symmetrizability is fundamental to the reduction result.   It is not needed in the general reduction result for linear variation \citep{Altenberg:1984,Altenberg:and:Feldman:1987,Altenberg:2009:Linear}.  In the absence of symmetrizability, the non-Perron eigenvalues may be complex.  Complex eigenvalues correspond to circulating non-zero net flows between states.  But as discussed in \Qref {Q3A}, net flows from fitter-than-average haplotypes to less-fit-than-average haplotypes are already a part of any mutation-selection balance.

 I conjecture that the symmetrizability constraint can be removed, and the multivariate reduction result will still pertain.

%%%%%%%%%%%%%%%%%%%%%%%%%%%%%
\subsection{The Multivariate Reduction Principle}
\label{subsec:Multivariate}

Details of the multivariate reduction principle are now discussed.

\subsubsection{Negative Correlations between Selection and Mutation Rates}
\label{subsubsec:Correlations}

A new feature of this model is that it analyzes modifier loci that individually tune the mutation rates different loci.  When the marginal fitnesses at equilibrium do not depend on a particular locus, the modifier locus can `detect' this, even in the midst of large complex fitness interactions among the other loci, by being able to change the mutation rates at this locus with no effect on the modifier allele's survival.  

This means if genetic variation exists for local mutation rates, these rates will evolve differently depending on whether the locus is neutral or not.  Empirical studies find substantial variation in mutation rates between sites within a genome \citep{Baer:Miyamoto:and:Denver:2007,King:and:Kashi:2007,Fox:Tuch:and:Chuang:2008}.  An implication of Theorem \ref{Theorem:MultivariateMutation} is that these differences may be the result different histories of selection among loci.  In particular, neutral loci do not have the reduction force operating on variation for their individual mutation rates, so they may evolve  higher mutation rates.  

If there were any mechanism that decreased mutations at one location at the expense of increasing it at another location, then neutral loci could become a `dumping ground' for such negative pleiotropic relations, and would enable their partner loci under selection to evolve lower mutation rates.  A potential example of such pleiotropic interactions is documented by \citet{Hoede:Denamur:and:Tenaillon:2006}, who find that single-stranded DNA secondary structure reduces mutation rates in \emph{E. coli}, and that such structures are found in excess within heavily transcribed sections of DNA.  If two sequences $A$ and $B$ were competing to form secondary structure with a third site $C$, modifier dynamics would favor the evolution of secondary structure to protect the sequence $A$ or $B$ incurring the greatest genetic load.

 The possibility that there is a systemic negative correlation between underlying mutation rates and selection intensity presents a confounding possibility for models of base substitution in phylogenetic models.

\subsubsection{The Neutral Manifold of Mutation Rates}
\label{subsubsec:Manifold}

The manifold, $\Nc(\muv)$, of mutation rates that are neutral for a new modifier allele is a topological necessity whenever all the loci individually exhibit the reduction result.  The finding in \citet{Zhivotovsky:Feldman:and:Christiansen:1994} that a weighted average of the recombination rates determines whether the new modifier allele increases or not is, in fact, the finding of the neutral manifold in the linear limit.  Their manifold can be defined by setting to zero their expression:
\eb
\label{eq:ZFCaverage}
\sum_{s=1}^L \sum_{\Above{t=1}{t \neq s}}^L A_{st} \frac{\rho_{st}}{r_{st}} = 0,
\ee
which produces an $L(L-1) - 1$ dimensional plane in the $L(L-1)$ dimensional space of pairwise recombination rates between $L$ loci.  Their manifold is a flat hyperplane, one may infer, as a consequence of the assumptions of weak selection with pairwise additive-by-additive epistasis, which eliminates many nonlinearities.  In the current paper, an explicit formula like \eqref {eq:ZFCaverage} for the neutral manifold never appears;  the existence and properties of this manifold are inferred through topological arguments, purely from the monotonicity and negative partial derivatives of the spectral radius $\partial \rho(\M_\muv \D) / \partial \mu_\kappa$.  

However, an explicit equation for the manifold can be given for small perturbations of $\muv$.  Let
\[
\dv  \eqdef \finalversion{\nabla \rho(\muv)} = \matrx{\pmu{ \rho(\M_\muv \D) }{\xi}}_{i=1}^L 
\]
refer to the \finalversion{gradient vector}.  Its value can be computed explicitly (numerically if not analytically) if $\M_\muv$ and $\D$ are given, using \eqref{eq:pmukSumForm} and \eqref{eq:xvh},
since $\Lam^{(\kappa)}$, $\Upsi_\muv$, $\K$, and $\B$ all derive from $\M_\muv$, and $\vvh$ and $ \rho(\M_\muv \D) $ derive from $\M_\muv$ and $\D$.

So for small $\del \in \Reals^L$, the manifold $ \Nc(\muv)$ is approximated near $\muv$ by
\eb
\label{eq:mudel}
\{\muv + \del \suchthat \dv^\tp \del = 0 \}.
\ee

The entries of $\dv$ are analogous to the weights $A_{st}$ in \eqref {eq:ZFCaverage}, which can be inferred to be  proportional to the derivatives of the spectral radius with respect to each recombination rate $r_{st}$. \citet{Zhivotovsky:Feldman:and:Christiansen:1994} point out that \eqref {eq:ZFCaverage} is not simply a total of all the changes in the recombination rates, but a weighted sum whose weights $A_{st}$ incorporate the intensity of epistasis between loci $s$ and $t$.   A simple sum would entail that the derivatives of the spectral radius be all equal, but clearly, the derivatives depend on selection and mutation or recombinations distributions in an intricate way.

A little reflection will show that what is found here and in \citet{Zhivotovsky:Feldman:and:Christiansen:1994} is, indeed,  the only possible form that a multivariate reduction principle could take.  A multivariate reduction principle should have, as its simplest requirement, that when a new modifier allele changes a single variable, it should increase if and only if it reduces the value of that variable.  \Secref{subsection:NeutralSurfaces} shows that this simple requirement leads, through  topological necessity, to the existence the neutral surface of mutation rates with its described properties.

%%%%%%%%%%%%%%
\subsection{The Strength of Selection on the Modifier Locus}
\label{subsec:Strength}

It is something of the `lore' about modifier genes that selection induced on them is weak, since a number of particular cases studied found slow changes in frequency of the modifier alleles (e.g. \citealt{Karlin:and:McGregor:1972:Modifier,Karlin:and:McGregor:1974}).  With weak selection and pairwise additive-by-additive epistasis, \citet{Zhivotovsky:Feldman:and:Christiansen:1994} find that selection induced on the modifier allele is quite small and the asymptotic rate of change in the modifier allele on the order of the square of the epistasis.  \citet{Kondrashov:1995} finds that the modifier alleles change frequency slowly in a model using various assumptions and approximations to estimate the selection induced on a mutation modifier in populations under mutation-selection balance.   

But small rates of change are not, in general, a necessity of modifier gene models.  As was shown in \citet[Result 2b]{Altenberg:and:Feldman:1987}, in the extreme case that $\etav = \0$, the asymptotic growth rate of the new modifier allele, $\rho(\M_\etav \D)$ will equal $\max_i \wh_i / \wbh > 1$, which has an upper bound of $1/\sigma$, where $\sigma$ is the fraction of haplotypes transmitted without change.

Here, $\sigma = \prod_{\xi=1}^L (1-\mu_\xi)$ is the fraction of haplotypes that are transmitted without any mutations.  If the number of loci is large, and the values of $\mu_\xi$ moderate, $\sigma$ can be quite small, and the upper bound $1/\sigma$ large.   In the Poisson approximation discussed in \Qref{Q4} above, part \itemref{QA:Poisson}, the estimate of $\sigma$ in humans is $\sigma = \Pr[0] = \lambda^\nu  e^{-\lambda} / \nu ! = 1.75^0 e^{-1.75} / 0! = 0.17$, so $1/\sigma = 5.7$.  Thus, the upper bound on the strength of induced selection coefficient of a mutation-eliminating modifier allele is around 6, which allows very strong induced selection on the modifier locus.  The actual value that $\max_i \wh_i / \wbh$ takes on depends on the specifics of the selection regime and mutation distributions and can be substantially less than $1/\sigma$.

\finalversion{The strength of selection on the new modifier can also be seen to increase with several factors:  the magnitude of its change on mutation rates (Theorem \ref {Theorem:Main}), the number of loci whose mutation rates it alters (Corollary \ref{Corollary:GlobalMutation}), and the number of cell divisions from zygote to gamete (Corollary \ref {Corollary:MultipleDivisions}). }

%%%%%%%%%%%%%%%%%%%%%%%%%%%%%
\subsection{Relation to Models that Depart from the Reduction Principle}
\label{subsec:Depart}

Departures from the reduction result in near-equilibrium populations have been found mainly in models that depart from linear variation.  The models here and in \citet{Zhivotovsky:Feldman:and:Christiansen:1994} have variation that is not linear, yet they both produce the multivariate reduction result.  What underlying properties can explain this?

As a way to summarize the examples of departures from the reduction principle, \citet[pp. 149, 225--228]{Altenberg:1984} proposed a `principle of partial control':   when the modifier gene has only partial control over the transformation occurring at loci under selection, then it may be possible for the part it controls to evolve an increase in rates.    I offered the following speculation:
\begin{quote}
If a modifier controls the transformation acting at only one or a
few loci, then the transformations acting at other loci will render the
variation at this modifier non-linear. It is conceivable, therefore,
that a modifier affecting recombination at only a few loci could evolve
to increase that recombination when recombination is occurring
elsewhere. \citep[p. 227]{Altenberg:1984}
\end{quote}
The above possibility is ruled out by the results in \citet{Zhivotovsky:Feldman:and:Christiansen:1994}, at least for weak selection and pairwise epistasis:  even when the modifier has only partial control over the recombination events --- because it varies only one or a few pairwise recombination frequencies --- it can only evolve to decrease the recombination rates below the neutral manifold.  And the same situation applies here for mutation rates:  any departures from the reduction result due to partial control over mutation rates are ruled out. 

One can speculate about what the underlying difference is between these models and the models that provided the basis for the principle of partial control, namely: recombination in the presence of mutation \citep{Feldman:Christiansen:and:Brooks:1980}, or migration \citep{Charlesworth:and:Charlesworth:1979}, or segregation and syngamy \citep{Charlesworth:Charlesworth:and:Strobeck:1979,Holsinger:and:Feldman:1983:LM}, or mutation in the presence of segregation and syngamy \citep{Holsinger:and:Feldman:1983:MM}.  Each of the latter models is a mixed process, in which the modifier locus controls one among multiple transformation processes that differ in their mathematical structure.  In the current paper, there is only one type of process --- mutation, and in \citet{Zhivotovsky:Feldman:and:Christiansen:1994}, there is only recombination.  Multiple instances of linear variation for only one kind of process are compounded together to produce the variation in transmission of the entire genome, which is nonlinear.

One may wonder whether it is the independent occurrence of multiple events that produces the reduction result.  Here, multiple mutations occur independently.  But the model of \citet{Zhivotovsky:Feldman:and:Christiansen:1994} does not assume that multiple recombination events are independent, and can accommodate arbitrary interference patterns.  So, the independence of events is not essential to the reduction results observed.

These two models of non-linear variation that preserve the reduction result do share the following:  they use multiple instances of homogeneous genetic processes to built up a multilocus, multivariate model.  Is this the key to their preservation of the reduction result?  Lest one surmise that it is the homogeneity of processes that is the underlying feature that produces the reduction result, the following model from \citet[pp. 149--151]{Altenberg:1984} provides a counterexample.  

The model posits a single locus upon which two different mutation processes act sequentially.  Each process is `House of Cards' mutation \citep{Kingman:1978,Kingman:1980}, where the mutation distribution matrix, $\Pm^{(i)}$, for each process $i$, is a rank one matrix:  $\Pm^{(i)} = \piv^{(i)} \ev^\tp$.  The modifier gene has linear control over one of the two processes, and varies either $\mu_1$ or $\mu_2$ in the expression,
\eb
\label{eq:Mmixed}
\M_{\mu_1, \mu_2} =  [ (1-\mu_1) \I + \mu_1 \piv^{(1)} \ev^\tp] [ (1-\mu_2) \I + \mu_2 \piv^{(2)} \ev^\tp] .
\ee
One can craft values for the variables that violate the reduction result: if $\piv^{(2)}$ is weighted towards the least fit haplotypes, while $\piv^{(1)}$ is weighted toward the most fit, and $\mu_1$ is small while $\mu_2$ is large, then a modifier which shifts its subpopulation toward the fitter haplotypes by increasing mutation rate $\mu_1$ will increase when rare, provided the variables are in the right ranges (e.g. for two alleles, $\wh_1 > \wh_2$, $\mu_1 = 0.1$, $\mu_2 = 0.4$, $\pi^{(1)}_1 = 0.9$, and $ \pi^{(2)}_1 = 0.1$).   

When these two processes $\Pm^{(1)}$ and $\Pm^{(2)}$ act on two {\em different loci}, however, they can no longer  interact in the same way.  The mutation matrix then becomes:
\eb
\label{eq:Msplit}
\M_{\mu_1, \mu_2} =  [ (1-\mu_1) \I \otimes \I + \mu_1 ( \piv^{(1)} \ev^\tp) \otimes \I] [ (1-\mu_2) \I \otimes \I + \mu_2 ( \I  \otimes \piv^{(2)} \ev^\tp) ] .
\ee
Indeed, since $\Pm^{(1)}$ and $\Pm^{(2)}$ are symmetrizable, \eqref {eq:Msplit} is simply an instance of \eqref{eq:Mmuv}, the model analyzed here, so the reduction results of Theorem \ref{Theorem:MultivariateMutation} apply.  So, we see that the reduction principle applies to nonlinear variation of the form \eqref {eq:Msplit}, but not of the form \eqref {eq:Mmixed}.  The only difference between them is that the two mutation process have a single target in  \eqref {eq:Mmixed}, but separate targets in \eqref {eq:Msplit}.

The picture that emerges is that when mixed processes are acting on the same set of loci, the expansion of one process can sometimes systematically shift the population toward the fitter genotypes, and cause modifiers that support this expansion to survive.  This is the essence of the deterministic mutation hypothesis for the evolution of sex and recombination \citep{Kondrashov:1982}.  The theoretical question then becomes, how do we identify which combinations of processes and conditions on selection will produce this effect?  

One can make a wild conjecture at this point:  that in all of the cases of modifier models where a mixing of forces produces departures from the reduction principle, then a `separation of forces' into linear variation on separate loci --- provided it is feasible to follow a form similar to going from \eqref {eq:Mmixed} to \eqref {eq:Msplit} --- will restore the reduction result.  Evaluation of this conjecture is deferred to future work.

%%%%%%%%%%%%
\section{Acknowledgements}

I thank Marc Feldman for his mentorship in the theory of modifier genes.  Anonymous reviewers provided pointed questions that were helpful in revising the paper.


\begin{thebibliography}{64}
\expandafter\ifx\csname natexlab\endcsname\relax\def\natexlab#1{#1}\fi
\expandafter\ifx\csname url\endcsname\relax
  \def\url#1{\texttt{#1}}\fi
\expandafter\ifx\csname urlprefix\endcsname\relax\def\urlprefix{URL }\fi

\bibitem[{Ababneh et~al.(2006)Ababneh, Jermiin, and
  Robinson}]{Ababneh:Jermiin:and:Robinson:2006}
Ababneh, F., Jermiin, L.~S., and Robinson, J.
\newblock 2006.
\newblock Generation of the exact distribution and simulation of matched
  nucleotide sequences on a phylogenetic tree.
\newblock Journal of Mathematical Modelling and Algorithms 5:291--308.
\newblock ISSN 1570-1166 (Print) 1572-9214 (Online).

\bibitem[{Altenberg(1984)}]{Altenberg:1984}
Altenberg, L., 1984.
\newblock A Generalization of Theory on the Evolution of Modifier Genes.
\newblock Ph.D. thesis, Stanford University.
\newblock Searchable online and available from University Microfilms, Ann
  Arbor, MI.

\bibitem[{Altenberg(2009)}]{Altenberg:2009:Linear}
Altenberg, L.
\newblock 2009.
\newblock The evolutionary reduction principle for linear variation in genetic
  transmission.
\newblock Bulletin of Mathematical Biology 71:1264--1284.

\bibitem[{Altenberg and Feldman(1987)}]{Altenberg:and:Feldman:1987}
Altenberg, L. and Feldman, M.~W.
\newblock 1987.
\newblock Selection, generalized transmission, and the evolution of modifier
  genes. {I}. {T}he reduction principle.
\newblock Genetics 117:559--572.

\bibitem[{Baer et~al.(2007)Baer, Miyamoto, and
  Denver}]{Baer:Miyamoto:and:Denver:2007}
Baer, C.~F., Miyamoto, M.~M., and Denver, D.~R.
\newblock 2007.
\newblock Mutation rate variation in multicellular eukaryotes: causes and
  consequences.
\newblock Nature Reviews Genetics 8:619--631.

\bibitem[{Balkau and Feldman(1973)}]{Balkau:and:Feldman:1973}
Balkau, B. and Feldman, M.~W.
\newblock 1973.
\newblock Selection for migration modification.
\newblock Genetics 74:171--174.

\bibitem[{Brandon(1982)}]{Brandon:1982}
Brandon, R.~N., 1982.
\newblock The levels of selection.
\newblock Pages 315--323 \emph{in} P.~Asquith and T.~Nickles, eds. PSA 1982,
  volume~1. Philosophy of Science Association, East Lansing, MI.

\bibitem[{Charlesworth(1990)}]{Charlesworth:1990:MSB}
Charlesworth, B.
\newblock 1990.
\newblock Mutation-selection balance and the evolutionary advantage of sex and
  recombination.
\newblock Genetical Research (Cambridge) 55:199--221.

\bibitem[{Charlesworth and
  Charlesworth(1979)}]{Charlesworth:and:Charlesworth:1979}
Charlesworth, B. and Charlesworth, D.
\newblock 1979.
\newblock Selection on recombination in clines.
\newblock Genetics 91:581--589.

\bibitem[{Charlesworth et~al.(1979)Charlesworth, Charlesworth, and
  Strobeck}]{Charlesworth:Charlesworth:and:Strobeck:1979}
Charlesworth, B., Charlesworth, D., and Strobeck, C.
\newblock 1979.
\newblock Selection for recombination in partially self-fertilizing
  populations.
\newblock Genetics 93:237--244.

\bibitem[{Deutsch and Neumann(1984)}]{Deutsch:and:Neumann:1984}
Deutsch, E. and Neumann, M.
\newblock 1984.
\newblock Derivatives of the {P}erron root at an essentially nonnegative matrix
  and the group inverse of an {M}-matrix.
\newblock Journal of Mathematical Analysis and Applications 102:1--29.

\bibitem[{Duistermaat and Kolk(2004)}]{Duistermaat:Kolk:2004}
Duistermaat, J.~J. and Kolk, J. A.~C., 2004.
\newblock Multidimensional Real Analysis {I}: {D}ifferentiation.
\newblock Number~86 in Cambridge Studies in Advanced Mathematics. Cambridge
  University Press.
\newblock ISBN 9780521551144.

\bibitem[{{Eyre-Walker} and Keightley(2007)}]{Eyre-Walker:and:Keightley:2007}
{Eyre-Walker}, A. and Keightley, P.~D.
\newblock 2007.
\newblock The distribution of fitness effects of new mutations.
\newblock Nature Reviews Genetics 8:610--618.

\bibitem[{Feldman(1972)}]{Feldman:1972}
Feldman, M.~W.
\newblock 1972.
\newblock Selection for linkage modification: {I}. {R}andom mating populations.
\newblock Theoretical Population Biology 3:324--346.

\bibitem[{Feldman and Balkau(1973)}]{Feldman:and:Balkau:1973}
Feldman, M.~W. and Balkau, B.
\newblock 1973.
\newblock Selection for linkage modification {II}. {A} recombination balance
  for neutral modifiers.
\newblock Genetics 74:713--726.

\bibitem[{Feldman et~al.(1980)Feldman, Christiansen, and
  Brooks}]{Feldman:Christiansen:and:Brooks:1980}
Feldman, M.~W., Christiansen, F.~B., and Brooks, L.~D.
\newblock 1980.
\newblock Evolution of recombination in a constant environment.
\newblock Proceedings of the National Academy of Sciences U.S.A. 77:4838--4841.

\bibitem[{Feldman and Krakauer(1976)}]{Feldman:and:Krakauer:1976}
Feldman, M.~W. and Krakauer, J., 1976.
\newblock Genetic modification and modifier polymorphisms.
\newblock Pages 547--583 \emph{in} S.~Karlin and E.~Nevo, eds. Population
  Genetics and Ecology. Academic Press, New York.

\bibitem[{Feldman and Liberman(1986)}]{Feldman:and:Liberman:1986}
Feldman, M.~W. and Liberman, U.
\newblock 1986.
\newblock An evolutionary reduction principle for genetic modifiers.
\newblock Proc. Natl. Acad. Sci. USA 83:4824--4827.

\bibitem[{Feller(1971)}]{Feller:1968v1}
Feller, W., 1971.
\newblock An Introduction to Probability Theory and Its Applications,
  volume~{I}.
\newblock John Wiley and Sons, New York, 3rd edition.

\bibitem[{Fox et~al.(2008)Fox, Tuch, and Chuang}]{Fox:Tuch:and:Chuang:2008}
Fox, A., Tuch, B., and Chuang, J.
\newblock 2008.
\newblock Measuring the prevalence of regional mutation rates: an analysis of
  silent substitutions in mammals, fungi, and insects.
\newblock BMC Evolutionary Biology 8:186.
\newblock ISSN 1471-2148.

\bibitem[{Giraud et~al.(2001)Giraud, Matic, Tenaillon, Clara, Radman, Fons, and
  Taddei}]{Giraud:Matic:etal:2001}
Giraud, A., Matic, I., Tenaillon, O., Clara, A., Radman, M., Fons, M., and
  Taddei, F.
\newblock 2001.
\newblock Costs and benefits of high mutation rates: Adaptive evolution of
  bacteria in the mouse gut.
\newblock Science 291.

\bibitem[{Guillemin and Pollack(1974)}]{Guillemin:and:Pollack:1974}
Guillemin, V. and Pollack, A., 1974.
\newblock Differential Topology.
\newblock Prentice-Hall, Englewood Cliffs, NJ.

\bibitem[{Hirsch(1976)}]{Hirsch:1976}
Hirsch, M.~W., 1976.
\newblock Differential Topology.
\newblock Springer-Verlag, New York.

\bibitem[{Hoede et~al.(2006)Hoede, Denamur, and
  Tenaillon}]{Hoede:Denamur:and:Tenaillon:2006}
Hoede, C., Denamur, E., and Tenaillon, O.
\newblock 2006.
\newblock Selection acts on dna secondary structures to decrease
  transcriptional mutagenesis.
\newblock PLoS Genetics 2:e176.
\newblock \urlprefix\url{http://dx.plos.org/10.1371%2Fjournal.pgen.0020176}.

\bibitem[{Holsinger et~al.(1986)Holsinger, Feldman, and
  Altenberg}]{Holsinger:Feldman:and:Altenberg:1986}
Holsinger, K., Feldman, M.~W., and Altenberg, L.
\newblock 1986.
\newblock Selection for increased mutation rates with fertility differences
  between matings.
\newblock Genetics 112:909--922.

\bibitem[{Holsinger and
  Feldman(1983{\natexlab{\emph{a}}})}]{Holsinger:and:Feldman:1983:LM}
Holsinger, K.~E. and Feldman, M.~W.
\newblock 1983{\natexlab{\emph{a}}}.
\newblock Linkage modification with mixed random mating and selfing: a
  numerical study.
\newblock Genetics 103:323--333.

\bibitem[{Holsinger and
  Feldman(1983{\natexlab{\emph{b}}})}]{Holsinger:and:Feldman:1983:MM}
Holsinger, K.~E. and Feldman, M.~W.
\newblock 1983{\natexlab{\emph{b}}}.
\newblock Modifiers of mutation rate: Evolutionary optimum with complete
  selfing.
\newblock Proceedings of the National Academy of Sciences U.S.A. 80:6732--6734.

\bibitem[{Horn and Johnson(1985)}]{Horn:and:Johnson:1985}
Horn, R.~A. and Johnson, C.~R., 1985.
\newblock Matrix Analysis.
\newblock Cambridge University Press, Cambridge.

\bibitem[{Iosifescu(1980)}]{Iosifescu:1980}
Iosifescu, M., 1980.
\newblock Finite Markov Processes and Their Applications.
\newblock John Wiley and Sons, Bucharest.

\bibitem[{Jayaswal et~al.(2005)Jayaswal, Jermiin, and
  Robinson}]{Jayaswal:Jermiin:and:Robinson:2005}
Jayaswal, V., Jermiin, L.~S., and Robinson, J.
\newblock 2005.
\newblock Estimation of phylogeny using a general markov model.
\newblock Evolutionary Bioinformatics Online 1:62--80.

\bibitem[{Karlin(1976)}]{Karlin:1976}
Karlin, S., 1976.
\newblock Population subdivision and selection migration interaction.
\newblock Pages 616--657 \emph{in} S.~Karlin and E.~Nevo, eds. Population
  Genetics and Ecology,. Academic Press, New York.

\bibitem[{Karlin(1982)}]{Karlin:1982}
Karlin, S., 1982.
\newblock Classification of selection-migration structures and conditions for a
  protected polymorphism.
\newblock Pages 61--204 \emph{in} M.~K. Hecht, B.~Wallace, and G.~T. Prance,
  eds. Evolutionary Biology, volume~14. Plenum Publishing Corporation.

\bibitem[{Karlin and
  McGregor(1972{\natexlab{\emph{a}}})}]{Karlin:and:McGregor:1972:Application}
Karlin, S. and McGregor, J.
\newblock 1972{\natexlab{\emph{a}}}.
\newblock Application of method of small parameters to multi-niche population
  genetic models.
\newblock Theoretical Population Biology 3:186--209.

\bibitem[{Karlin and
  McGregor(1972{\natexlab{\emph{b}}})}]{Karlin:and:McGregor:1972:Modifier}
Karlin, S. and McGregor, J.
\newblock 1972{\natexlab{\emph{b}}}.
\newblock The evolutionary development of modifier genes.
\newblock Proceedings of the National Academy of Sciences U.S.A. 69:3611--3614.

\bibitem[{Karlin and McGregor(1974)}]{Karlin:and:McGregor:1974}
Karlin, S. and McGregor, J.
\newblock 1974.
\newblock Towards a theory of the evolution of modifier genes.
\newblock Theoretical Population Biology 5:59--103.

\bibitem[{Keilson(1979)}]{Keilson:1979}
Keilson, J., 1979.
\newblock Markov Chain Models: Rarity and Exponentiality.
\newblock Springer-Verlag, New York.

\bibitem[{King and Kashi(2007)}]{King:and:Kashi:2007}
King, D.~G. and Kashi, Y.
\newblock 2007.
\newblock Mutation rate variation in eukaryotes: evolutionary implications of
  site-specific mechanisms.
\newblock Nature Reviews Genetics 8.

\bibitem[{Kingman(1978)}]{Kingman:1978}
Kingman, J. F.~C.
\newblock 1978.
\newblock A simple model for the balance between selection and mutation.
\newblock Journal of Applied Probability 15:1--12.

\bibitem[{Kingman(1980)}]{Kingman:1980}
Kingman, J. F.~C., 1980.
\newblock Mathematics of Genetic Diversity.
\newblock Society for Industrial and Applied Mathematics, Philadelphia.
\newblock ISBN 0-89871-166-5.

\bibitem[{Kondrashov(1982)}]{Kondrashov:1982}
Kondrashov, A.~S.
\newblock 1982.
\newblock Selection against harmful mutations in large sexual and asexual
  populations.
\newblock Genetical Research (Cambridge) 40::325--332.

\bibitem[{Kondrashov(1984)}]{Kondrashov:1984}
Kondrashov, A.~S.
\newblock 1984.
\newblock Deleterious mutations as an evolutionary factor. {I}. {T}he advantage
  of recombination.
\newblock Genetical Research (Cambridge) 44:199--217.

\bibitem[{Kondrashov(1995)}]{Kondrashov:1995}
Kondrashov, A.~S.
\newblock 1995.
\newblock Modifiers of mutation-selection balance: general approach and the
  evolution of mutation rates.
\newblock Genetical Research 66:53--69.

\bibitem[{Kondrashov and Kondrashov(2010)}]{Kondrashov2:2010}
Kondrashov, F.~A. and Kondrashov, A.~S.
\newblock 2010.
\newblock Measurements of spontaneous rates of mutations in the recent past and
  the near future.
\newblock Philosophical Transactions of the Royal Society {B} 365:1169--1176.

\bibitem[{Lewontin(1974)}]{Lewontin:1974}
Lewontin, R.~C., 1974.
\newblock The Genetic Basis of Evolutionary Change.
\newblock Columbia University Press.

\bibitem[{Liberman and
  Feldman(1986{\natexlab{\emph{a}}})}]{Liberman:and:Feldman:1986:GRP}
Liberman, U. and Feldman, M.~W.
\newblock 1986{\natexlab{\emph{a}}}.
\newblock A general reduction principle for genetic modifiers of recombination.
\newblock Theoretical Population Biology 30:341--371.

\bibitem[{Liberman and
  Feldman(1986{\natexlab{\emph{b}}})}]{Liberman:and:Feldman:1986:MMR}
Liberman, U. and Feldman, M.~W.
\newblock 1986{\natexlab{\emph{b}}}.
\newblock Modifiers of mutation rate: {A} general reduction principle.
\newblock Theoretical Population Biology 30:125--142.

\bibitem[{Lynch(2010)}]{Lynch:2010:Rate}
Lynch, M.
\newblock 2010.
\newblock Rate, molecular spectrum, and consequences of human mutation.
\newblock Proceedings of the National Academy of Sciences U.S.A. 107:961--968.

\bibitem[{Lynch et~al.(2008)Lynch, Sung, Morris, Coffey, Landry, Dopman,
  Dickinson, Okamoto, Kulkarni, Hartl, and Thomas}]{Lynch:Sung:etal:2008}
Lynch, M., Sung, W., Morris, K., Coffey, N., Landry, C.~R., Dopman, E.~B.,
  Dickinson, W.~J., Okamoto, K., Kulkarni, S., Hartl, D.~L., and Thomas, W.~K.
\newblock 2008.
\newblock {A genome-wide view of the spectrum of spontaneous mutations in
  yeast}.
\newblock Proceedings of the National Academy of Sciences 105:9272--9277.

\bibitem[{Munkres(1975)}]{Munkres:1975}
Munkres, J.~R., 1975.
\newblock Topology: A First Course.
\newblock Prentice-Hall, Englewood Cliffs, NJ.
\newblock ISBN 0-13-925495-1.

\bibitem[{Otto and Feldman(1997)}]{Otto:and:Feldman:1997}
Otto, S.~P. and Feldman, M.~W.
\newblock 1997.
\newblock Deleterious mutations, variable epistatic interactions, and the
  evolution of recombination.
\newblock Theoretical Population Biology 51:34--47.

\bibitem[{Pylkov et~al.(1998)Pylkov, Zhivotovsky, and
  Feldman}]{Pylkov:Zhivotovsky:and:Feldman:1998}
Pylkov, K.~V., Zhivotovsky, L.~A., and Feldman, M.~W.
\newblock 1998.
\newblock Migration versus mutation in the evolution of recombination under
  multilocus selection.
\newblock Genetical Research (Cambridge) 71:247--256.

\bibitem[{Roach et~al.(2010)Roach, Glusman, Smit, Huff, Hubley, Shannon, Rowen,
  Pant, Goodman, Bamshad, Shendure, Drmanac, Jorde, Hood, and
  Galas}]{Roach:Glusman:etal:2010}
Roach, J.~C., Glusman, G., Smit, A. F.~A., Huff, C.~D., Hubley, R., Shannon,
  P.~T., Rowen, L., Pant, K.~P., Goodman, N., Bamshad, M., Shendure, J.,
  Drmanac, R., Jorde, L.~B., Hood, L., and Galas, D.~J.
\newblock 2010.
\newblock Analysis of genetic inheritance in a family quartet by whole-genome
  sequencing.
\newblock Science Pages science.1186802+.
\newblock ISSN 1095-9203.
\newblock \urlprefix\url{http://dx.doi.org/10.1126/science.1186802}.

\bibitem[{Rodr\'{i}guez et~al.(1990)Rodr\'{i}guez, Oliver, Mar\'{i}n, and
  Medina}]{Rodriguez:et:al:1990}
Rodr\'{i}guez, F., Oliver, J., Mar\'{i}n, A., and Medina, J.
\newblock 1990.
\newblock The general stochastic model of nucleotide substitution.
\newblock Journal of Theoretical Biology 142:485 -- 501.
\newblock ISSN 0022-5193.

\bibitem[{Salmon(1971)}]{Salmon:1971}
Salmon, W.~C., 1971.
\newblock Statistical Explanation and Statistical Relevance.
\newblock University of Pittsburgh Press, Pittsburgh.

\bibitem[{Salmon(1984)}]{Salmon:1984}
Salmon, W.~C., 1984.
\newblock Scientific Explanation and the Causal Structure of the World.
\newblock Princeton University Press, Princeton, N. J.

\bibitem[{Singer and Thorpe(1967)}]{Singer:and:Thorpe:1967}
Singer, I.~M. and Thorpe, J.~A., 1967.
\newblock Lecture Notes on Elementary Topology and Geometry.
\newblock Springer-Verlag, New York.
\newblock ISBN 0-387-90202-3.

\bibitem[{Squartini and Arndt(2008)}]{Squartini:and:Arndt:2008}
Squartini, F. and Arndt, P.~F.
\newblock 2008.
\newblock Quantifying the stationarity and time reversibility of the nucleotide
  substitution process.
\newblock Molecular Biology and Evolution 25:2525--2535.

\bibitem[{Teague(1977)}]{Teague:1977}
Teague, R.
\newblock 1977.
\newblock A model of migration modification.
\newblock Theoretical Population Biology 12:86--94.

\bibitem[{Whelan and Goldman(2004)}]{Whelan:and:Goldman:2004}
Whelan, S. and Goldman, N.
\newblock 2004.
\newblock Estimating the frequency of events that cause multiple-nucleotide
  changes.
\newblock Genetics 167:2027--2043.

\bibitem[{Wilkinson(1965)}]{Wilkinson:1965}
Wilkinson, J.~H., 1965.
\newblock The Algebraic Eigenvalue Problem.
\newblock Clarendon Press, Oxford.

\bibitem[{Yang(1995)}]{Yang:1995}
Yang, Z.
\newblock 1995.
\newblock On the general reversible markov process model of nucleotide
  substitution: A reply to saccone et al.
\newblock Journal of Molecular Evolution 41:254--255.

\bibitem[{Yang and Nielsen(2002)}]{Yang:and:Nielsen:2002}
Yang, Z. and Nielsen, R.
\newblock 2002.
\newblock Codon-substitution models for detecting molecular adaptation at
  individual sites along specific lineages.
\newblock Molecular Biology and Evolution 19:908--917.

\bibitem[{Zhivotovsky and Feldman(1995)}]{Zhivotovsky:and:Feldman:1995}
Zhivotovsky, L.~A. and Feldman, M.~W.
\newblock 1995.
\newblock The reduction principle for recombination under density-dependent
  selection.
\newblock Theoretical Population Biology 47:244--256.

\bibitem[{Zhivotovsky et~al.(1994)Zhivotovsky, Feldman, and
  Christiansen}]{Zhivotovsky:Feldman:and:Christiansen:1994}
Zhivotovsky, L.~A., Feldman, M.~W., and Christiansen, F.~B.
\newblock 1994.
\newblock Evolution of recombination among multiple selected loci: {A}
  generalized reduction principle.
\newblock Proceedings of the National Academy of Sciences U.S.A. 91:1079--1083.

\end{thebibliography}
\end{document}